%% file: main_arxiv_dl23.tex
\documentclass{ceurart}

\usepackage{listings}
\usepackage{graphbox}
\DeclareRobustCommand{\DLLogo}{%
  \begingroup\normalfont
  \kern-1.75pt\includegraphics[align=c,height=1.25\baselineskip]{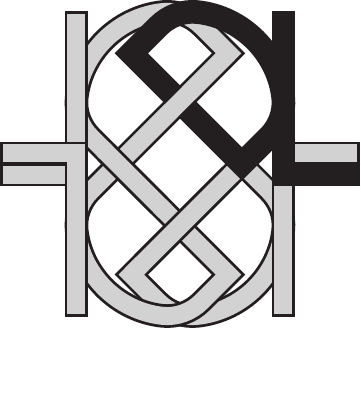}\kern-1.5pt%
  \endgroup
}

\usepackage{paralist}
\usepackage[utf8]{inputenc}
\usepackage[T1]{fontenc}
\usepackage{amsfonts}
\usepackage{latexsym, amsmath, amssymb, amsthm}
\usepackage{stmaryrd}
\usepackage{wasysym}
\usepackage{bbm}
\usepackage{rotating}
\usepackage{pifont}
\usepackage{booktabs}
\usepackage{array}
\usepackage{multirow}
\usepackage{tikz}
\usepackage{color}
\usepackage{xcolor}
\usepackage[inline]{enumitem}
\usepackage{xspace}
\usepackage{thm-restate}
\usepackage{phonetic}
\usepackage[ruled,vlined]{algorithm2e}
\usepackage{hyperref}

\newtheorem{theorem}{Theorem}
\newtheorem{proposition}{Proposition}
\newtheorem{definition}{Definition}

\newtheorem{lemma}{Lemma}
\newtheorem{corollary}{Corollary}
\newtheorem{claim}{Claim}

\newtheorem{remark}{Remark}

\newif\ifappendix
\appendixtrue

\newif\ifshort
\shortfalse

\input{macros}
\begin{document}

\copyrightyear{2023}
\copyrightclause{Copyright for this paper by its authors.
  Use permitted under Creative Commons License Attribution 4.0
  International (CC BY 4.0).}

\conference{
}

\ifappendix
\title{Non-Rigid Designators in Epistemic and Temporal Free Description Logics (Extended Version)}
\fi

\author[1]{Alessandro Artale}
[
orcid=0000-0002-3852-9351,
email=artale@inf.unibz.it,
url=http://www.inf.unibz.it/~artale/,
]
\address[1]{Free University of Bozen-Bolzano}

\author[2]{Andrea Mazzullo}
[
orcid=0000-0001-8512-1933,
email=andrea.mazzullo@unitn.it,
url=https://sites.google.com/view/andreamazzullo/,
]
\address[2]{University of Trento}

%

%


\input{0-abstract.tex}




\maketitle


\input{1-introduction.tex}

\input{2-preliminaries.tex}

\input{3-motivations.tex}

\input{4-ml-reasoning.tex}

\input{5-time-reasoning.tex}

\input{6-conclusion.tex}

\bibliography{bibliography}

\end{document}

%% file: macros.tex
\newcommand{\nb}[1]{}

\newcommand{\p}{\varphi}
\newcommand{\q}{\psi}

\newcommand{\eset}{\emptyset}

\newcommand{\Int}{\ensuremath{\mathcal{I}}\xspace}

\newcommand{\con}[1]{\ensuremath{\mathsf{con}(#1)}\xspace}

\newcommand{\Until}{\ensuremath{\mathbin{\mathcal{U}}}\xspace}

\newcommand{\Next}{{\ensuremath{\raisebox{0.25ex}{\text{\scriptsize{$\bigcirc$}}}}}}

\newcommand{\monodic}{\ensuremath{\mathop{\ooalign{$\Box$ \cr \kern0.57ex \raisebox{0.2ex}{\scalebox{0.55}{$1$}}}\rule{0pt}{1.5ex} \kern-0.7ex}}\xspace}

\DeclareRobustCommand{\zerodic}{\ensuremath{\mathop{\ooalign{$\Box$ \cr \kern0.57ex \raisebox{0.2ex}{\scalebox{0.55}{$0$}}}\rule{0pt}{1.5ex} \kern-0.7ex}}\xspace}

\newcommand{\NC}{\ensuremath{{\sf N_C}}\xspace}
\newcommand{\NI}{\ensuremath{{\sf N_I}}\xspace}

\newcommand{\NR}{\ensuremath{{\sf N_R}}\xspace}

\newcommand{\ALCO}{\ensuremath{\smash{\mathcal{ALCO}}}\xspace}

\newcommand{\ALCOu}{\ensuremath{\smash{\mathcal{ALCO}_{\!u}}}\xspace}

\newcommand{\ALCOud}{\ensuremath{\smash{\mathcal{ALCO}_{\!u}^{\defdes}}}\xspace}

\newcommand{\ELO}{\ensuremath{\smash{\mathcal{ELO}}}\xspace}

\newcommand{\ELOud}{\ensuremath{\smash{\mathcal{ELO}_{\!u}^{\defdes}}}\xspace}

\newcommand{\quasimod}{\ensuremath{\Qmf}\xspace}

\newcommand{\runs}{\ensuremath{\Rmf}\xspace}
\newcommand{\gruns}{\ensuremath{\Gmf}\xspace}

\newcommand{\funcand}{\ensuremath{\boldsymbol{q}}\xspace}

\newcommand{\settp}{\ensuremath{\boldsymbol{T}}\xspace}

\newcommand{\contp}{\ensuremath{{\boldsymbol{t}}}\xspace}

\newcommand{\ML}{\ensuremath{\smash{\mathcal{ML}}}\xspace}

\newcommand{\MLALCOu}{\ensuremath{\smash{\ML_{\ALCOu}}}\xspace}
\newcommand{\MLALCOud}{\ensuremath{\smash{\ML_{\ALCOud}}}\xspace}

\newcommand{\Sfivee}{\ensuremath{\mathbf{S5}\times\mathbf{S5}}}

\newcommand{\SfiveALC}{\ensuremath{\mathbf{S5}_{\mathcal{ALC}}}}

\newcommand{\SfiveALCOu}{\ensuremath{\mathbf{S5}_{\ALCOu}}}
\newcommand{\SfiveALCOud}{\ensuremath{\mathbf{S5}_{\ALCOud}}}

\newcommand{\TL}{\ensuremath{\smash{\mathcal{TL}}}\xspace}

\newcommand{\TLALCOud}{\ensuremath{\smash{\TL_{\mathcal{ALCO}_{\!u}^{\defdes}}}}\xspace}
\newcommand{\TLALCOu}{\ensuremath{\smash{\TL_{\mathcal{ALCO}_{\!u}}}}\xspace}

\newcommand{\LTLALCO}{\ensuremath{\smash{\LTL_{\mathcal{ALCO}}}}\xspace}

\newcommand{\LTLALCOd}{\ensuremath{\smash{\LTL_{\mathcal{ALCO}^{\defdes}}}}\xspace}
\newcommand{\LTLALCOud}{\ensuremath{\smash{\LTL_{\mathcal{ALCO}_{\!u}^{\defdes}}}}\xspace}
\newcommand{\LTLALCOu}{\ensuremath{\smash{\LTL_{\mathcal{ALCO}_{\!u}}}}\xspace}

\newcommand{\LTLfALCOud}{\ensuremath{\smash{\LTL^{\boldsymbol{f}}_{\mathcal{ALCO}_{\!u}^{\defdes}}}}\xspace}
\newcommand{\LTLfALCOu}{\ensuremath{\smash{\LTL^{\boldsymbol{f}}_{\mathcal{ALCO}_{\!u}}}}\xspace}

\newcommand{\LTL}{\ensuremath{\mathbf{LTL}}\xspace}

\newcommand{\QTLfr}[3]{\ensuremath{{\mathcal{T}_{\Until}\mathcal{QL}^{#2}_{#3}}}\xspace}

\newcommand{\TDLite}{\ensuremath{\textit{TDL-Lite}}\xspace}

\newcommand{\last}{\ensuremath{\textit{last}}\xspace}

\newcommand{\defdes}{\ensuremath{\smash{\iota}}\xspace}

\newcommand{\PTime}{\textsc{PTime}}

\newcommand{\ExpTime}{\textsc{ExpTime}}
\newcommand{\NExpTime}{\textsc{NExpTime}}

\newcommand{\Fmf}{\ensuremath{\mathfrak{F}}\xspace}
\newcommand{\Gmf}{\ensuremath{\mathfrak{G}}\xspace}

\newcommand{\Mmf}{\ensuremath{\mathfrak{M}}\xspace}

\newcommand{\Qmf}{\ensuremath{\mathfrak{Q}}\xspace}
\newcommand{\Rmf}{\ensuremath{\mathfrak{R}}\xspace}

\newcommand{\Tmf}{\ensuremath{\mathfrak{T}}\xspace}

\newcommand{\Cmc}{\ensuremath{\mathcal{C}}\xspace}

\newcommand{\Imc}{\ensuremath{\mathcal{I}}\xspace}

%% file: 0-abstract.tex
\begin{abstract}
Definite descriptions, such as `the smallest planet in the Solar System', have been recently recognised as semantically transparent devices for object identification in knowledge representation formalisms.
Along with individual names, they have been introduced also in the context of description logic languages, enriching the expressivity of standard nominal constructors.
Moreover, in the first-order modal logic literature, definite descriptions have been widely investigated for their \emph{non-rigid} behaviour, which allows them to denote different objects at different states.
In this direction, we introduce \emph{epistemic} and \emph{temporal} extensions of standard description logics, with nominals and the universal role, additionally equipped with definite descriptions constructors.
Regarding names and descriptions, in these languages we allow for: possible \emph{lack of denotation}, ensured by \emph{partial} models, coming from free logic semantics as a generalisation of the classical ones; and \emph{non-rigid designation} features, obtained by assigning to terms distinct values across states, as opposed to the standard rigidity condition on individual expressions.
In the absence of the rigid designator assumption, we show that the satisfiability problem for epistemic free description logics is $\NExpTime$-complete, while satisfiability for temporal free description logics
over linear time structures
is undecidable.
\end{abstract}

%% file: 1-introduction.tex
\section{Introduction}
\label{sec:intro}

\emph{Definite descriptions},
like `the smallest planet in the Solar System',
are expressions having form `the $x$ such that $\p$'. Together with \emph{individual names},
such as `Mercury',
they are used as \emph{referring expressions} to identify objects in a given domain~\cite{BorEtAl16,BorEtAl17,TomWed18}.
Definite description and individual names can also \emph{fail to denote} any object at all, as in the cases of the definite description `the planet between Mercury and the Sun' or the individual name `Vulcan'.
Formal accounts that address these aspects and still admit definite descriptions as genuine terms of the language, on a par with individual names, are usually based on so-called \emph{free logics}~\cite{Ben02,Leh02,Ind21,IndZaw21}.
These are in contrast with classical logic approaches, in which individual names are assumed to always designate, and where definite descriptions are paraphrased in terms of sentences expressing existence and uniqueness conditions (an approach dating back to Russell~\cite{Rus05}).
\ifappendix
Recently,
definite descriptions have been introduced into description logic (DL) formalisms~\cite{NeuEtAl20,ArtEtAl20b,ArtEtAl21a}, by enriching standard languages with nominals like $\ALCO$ and $\ELO$.
In particular, the DLs $\ALCOud$ and $\ELOud$ from~\cite{ArtEtAl20b,ArtEtAl21a}, which include the universal role and are based on a free logic semantics allowing for non-denoting terms, have been shown to have, respectively, $\ExpTime$-complete ontology satisfiability and $\PTime$-complete subsumption checking problem, hence matching the complexities of the classical
DL counterparts.
\fi

\ifappendix
In addition, expressions referring to the same object cannot always be mutually substituted  in a sentence while preserving the truth value.
For instance, even if everybody is aware that Mercury is very close to the Sun, not everyone thereby knows that also the smallest planet in the Solar System is very close to the Sun, despite the fact that `the smallest planet in the Solar System' and `Mercury' denote the same object.
Indeed, \emph{modal} settings, like the temporal or the epistemic ones, are
\emph{referentially opaque} contexts, where the \emph{intension} (i.e., the meaning) of a term might not coincide with its \emph{extension} (that is, its referent)~\cite{Fit04}.
In these cases, referring expressions can behave as \emph{non-rigid designators}, meaning that they can denote different individuals across different states (epistemic alternatives, instants of time, etc.).

For example, in an epistemic scenario, `the smallest planet' can be conceived as denoting another astronomical object, by someone unaware of its actual reference to Mercury. Similarly, on a temporal scale, `the smallest planet' refers to Mercury nowadays, but denoted Pluto less than twenty years ago~\cite{FitMen12}.
%
Due to this problematic interplay between designation and modalities,
non-rigid descriptions and names
have been widely investigated in first-order modal and temporal logics~\cite{Coc84,Gar01,BraGhi07,KroMerz08,FitMen12,CorOrl13,Ind20,Orl21},
as \emph{individual concepts} or \emph{flexible terms}
capable of taking different values across states.
\fi
%
\ifappendix
However, despite the wide body of research both on temporal~\cite{WolZak98,ArtFra05,LutEtAl08} and epistemic~\cite{DonEtAl98,CalEtAl08,ConLen20} extensions,
non-rigid designators have received, to the best of our knowledge,
little attention in modal DLs.  In an epistemic DL context, non-rigid
individual names appear in~\cite{MehRud11}, under an approach that
involves abstract individual names interpreted on an infinite common
domain, but without definite descriptions.  
\fi

In this paper, we extend the free DLs proposed for the non-modal case
in~\cite{ArtEtAl20b,ArtEtAl21a}, by: (i) adding \emph{epistemic
  modalities}, such as $\Box$ (\emph{box}, read as `it is known
that'), or \emph{temporal} ones, like $\Until$ (\emph{until})\ifappendix applicable both to formulas and concepts of the
language; \fi (ii) introducing nominals built from \emph{definite
  descriptions} of the form $\iota C$ (read as `the object that is
$C$'), where $C$ is a concept, alongside the standard ones based on
\emph{individual names};
(iii) dropping the \emph{rigid designator assumption}, hence allowing
terms to behave as flexible individual concepts across states.  We
study the complexity of formula satisfiability, showing that, without
the rigid designator assumption, this problem for epistemic free DLs
is $\NExpTime$-complete (same as the logic \Sfivee~\cite[Theorem
5.26]{GabEtAl03}), whereas it becomes undecidable for temporal free
DLs interpreted on linear time structures (while\nb{A: added} it is
decidable without definite descriptions and with the RDA~\cite[Theorem
14.12]{GabEtAl03}).

Section~\ref{sec:tdlnonrig} provides the necessary background on the
epistemic and temporal free DLs introduced in this paper. In
Section~\ref{sec:motivation}, we motivate with examples the
syntactical and semantical choices for these
languages. Section~\ref{sec-reasoning} studies the complexity of
formula satisfiability in epistemic free DLs, while
Section~\ref{sec:reasontfdl} focuses on the undecidability of temporal
free DLs over time flows consisting of finite or infinite traces.
Finally, Section~\ref{sec:conc} concludes the paper, discussing open
problems and future research directions.


%% file: 2-preliminaries.tex
\section{Epistemic and Temporal Free Description Logics}
\label{sec:tdlnonrig}

\ifappendix
In this section, we introduce the syntax of
$\MLALCOud$ and $\TLALCOud$, respectively a modal and a temporal
extension of the free DL $\ALCOud$ language~\cite{ArtEtAl20b,ArtEtAl21a}, as well as their semantics based on epistemic and temporal frames, respectively.

\subsection{Epistemic free description logics}
\label{subsec:ml-nonrig}
\fi

\ifappendix
The
$\MLALCOud$
language
is a modalised extension of the free DL
$\ALCOud$
language.
\fi
%
\ifappendix
Let \NC, \NR and \NI be countably infinite and pairwise disjoint sets
of \emph{concept names}, \emph{role names}, and \emph{individual
  names}, respectively. The $\MLALCOud$ \emph{terms} and
\emph{concepts} are defined as:
\begin{gather*}
  \tau ::= a \mid \defdes C, \qquad C ::= A \mid \{ \tau \} \mid \lnot
  C \mid C \sqcap C
\mid \exists r.C
\mid \exists u.C
\mid \Diamond C
\end{gather*}
where $a \in \NI$, $A \in \NC$, $r \in \NR$,
and $u$ is the
\emph{universal role}.
A term of the form $\defdes C$ is called a \emph{definite
  description}, with the concept $C$ being the \emph{body of
  $\defdes C$}, and a concept $\{ \tau \}$ is called a \emph{(term)
  nominal}. All the usual syntactic abbreviations are assumed:
$\bot = A \sqcap \lnot A$, $\top = \lnot \bot$,
$C \sqcup D = \lnot (\lnot C \sqcap \lnot D)$,
$C \Rightarrow D = \lnot C \sqcup D$,
$\forall s. C = \lnot \exists s. \lnot C$, with
$s \in \NR \cup \{ u \}$, and $\Box C = \lnot \Diamond \lnot C$. We
will consider also the \emph{reflexive diamond} operator,
$\Diamond^{+} C = C \sqcup \Diamond C$, and the \emph{reflexive box}
operator, $\Box^{+} C = C \sqcap \Box C$.

An \emph{$\MLALCOud$ axiom}, denoted as $\alpha$, is either a
\emph{concept inclusion} (\emph{CI}) of the form $C \sqsubseteq D$,
or an \emph{$\MLALCOud$ assertion} of the form $C(\tau)$ or
$r(\tau_1,\tau_2)$, where $C, D$ are concepts,  $r \in \NR$, and
$\tau, \tau_1, \tau_2$ are terms.
An \emph{$\MLALCOud$ formula} is an expression of the following form:
\[
\p ::=
(\alpha)
\mid \neg \p \mid \p \land \p
\mid \Diamond \p
\]
\fi

\ifappendix
Given a \emph{frame} $\Fmf = (W, \lhd)$, with $W$ being a non-empty
set of \emph{worlds} (or \emph{states}) and
$\lhd \subseteq W \times W$ being a binary \emph{accessibility
  relation} on $W$, a \emph{partial modal interpretation} based on
$\Fmf$ is a triple $\Mmf = ( \Fmf, \Delta, \Int)$, where: $\Fmf$ is
the frame of $\Mmf$; $\Delta$ is a non-empty set, called the
\emph{domain} of $\Mmf$ (we adopt the so-called \emph{constant domain
  assumption}~\cite{GabEtAl03}); and $\Int$ is a function associating
with every $w \in W$ a \emph{partial interpretation}
$\Imc_{w} = (\Delta, \cdot^{\Imc_{w}})$ that maps every $A \in \NC$ to
a subset of $\Delta$, every $r\in\NR$ to a subset of
$\Delta \times \Delta$, the universal role $u$ to the set
$\Delta \times \Delta$ itself, and every $a$ in a \emph{subset} of
$\NI$ to an element in $\Delta$. In other words, every
$\cdot^{\Imc_{w}}$ is a total function on $\NC \cup \NR$ and a
\emph{partia}l function on $\NI$.
We say that
$\Mmf = ( \Fmf, \Delta, \Int)$ is a \emph{total modal
  interpretation} if every $\Imc_{w}$, with $w \in W$, is a \emph{total}
interpretation,
meaning that $\cdot^{\Imc_{w}}$ is defined as above, except that it maps \emph{every} $a \in \NI$ to an element of $\Delta$.
\fi

\ifappendix
\begin{definition}
  Given $\Mmf = ( \Fmf, \Delta, \Int)$, with $\Fmf = (W, \lhd)$, we
  say that $\Mmf$ satisfies the \emph{rigid designator assumption}
  (\emph{RDA}) if, for every individual name $a \in \NI$ and every
  $w, v \in W$,
    the following condition holds:
if $a^{\Int_{w}}$ is defined, then
    $a^{\Int_{w}} = a^{\Int_{v}}$, i.e., $a$ is a \emph{rigid
      designator}.
  %
  An individual name $a \in \NI$ is said
  to \emph{denote in $\Int_{w}$} if $a^{\Int_{w}}$ is defined,
  and we say that it \emph{denotes in $\Mmf$} if $a$ denotes in $\Imc_{w}$, for some $w \in W$.
  Moreover, $a$ is called a \emph{ghost in} $\Mmf$ if,
  for every $w\in W$, $a$ does not denote in $\Int_{w}$.
\end{definition}
\begin{remark}
Dropping the RDA is the most general
  assumption, since rigid designators can be enforced by the CI
  $\Diamond^{+}\{a\} \sqsubseteq \Box^+ \{a\}$. Moreover, partial interpretations are a generalisation of the classical ones: an individual
  can be forced to denote at some state (i.e., not being a ghost) with the CI
  $\top \sqsubseteq \Diamond^+\exists u.\{a\}$, and to denote at all states with the formula $\Box^+ (\top \sqsubseteq \exists u.\{a\})$.
  Note that a ghost individual is vacuously rigid.
\end{remark}
\fi

Given
\ifappendix
$\Mmf = (\Fmf, \Delta, \Int)$, with $\Fmf = (W, \lhd)$,
\fi
and a
world $w \in W$,
we
define the \emph{value} $\tau^{\Int_{w}}$ of a term $\tau$ in $w$ as $a^{\Int_{w}}$, if $\tau  = a$, and as follows, for $\tau = \defdes C$:
\begin{gather*}
		(\defdes C)^{\Int_{w}}  =
			\begin{cases}
				d, & \text{if} \ C^{\Int_{w}} = \{ d \}, \ \text{for some} \ d \in \Delta; \\
				\text{undefined}, & \text{otherwise}.
			\end{cases}
\end{gather*}
As for the \emph{extension} of a concept $C$ in $w$, $C^{\Int_{w}}$ is
as usual with the following additions:
\begin{align*}
  (\Diamond C)^{\Imc_{w}} &= \{ d \in \Delta \mid \exists
  v \in W,
  \ifappendix
  w \lhd v \colon d \in C^{\Imc_{v}}
  \fi
  \},
  \quad
  \{ \tau \}^{\Imc_{w}} &=
			\begin{cases}
				\{ \tau^{\Imc_{w}} \}, & \text{if $\tau$ denotes in $\Imc_{w}$}, \\
				\, \emptyset, & \text{otherwise,}
			\end{cases}
\end{align*}
where a term $\tau$ is said to \emph{denote} in $\Imc_{w}$ if
$\tau^{\Imc_{w}}$ is defined.
A concept $C$ is \emph{satisfied at $w$ of $\Mmf$} if
$C^{\Imc_{w}} \neq \eset$. 
%
An
\ifappendix
\emph{$\MLALCOud$ formula $\p$ is satisfied at $w$ of $\Mmf$},
\fi
  written $\Mmf, w \models \p$, when:
%
\ifappendix
\begin{align*}
	\Mmf, w \models C\sqsubseteq D  &\text{ \, iff \, } C^{\Imc_{w}} \subseteq D^{\Imc_{w}}, \\
	\Mmf, w  \models C(\tau) &\text{ \, iff \, } \tau ~\text{denotes in}~{\Imc_{w}} ~\text{and}~\tau^{\Imc_{w}}\in C^{\Imc_{w}},\\
	\Mmf, w \models r(\tau_1,\tau_2) &\text{ \, iff \, }
        \tau_1,\tau_2 ~\text{denotes in}~{\Imc_{w}} ~\text{and}~(\tau_1^{\Imc_{w}},\tau_2^{\Imc_{w}}) \in r^{\Imc_{w}}, \\
	\Mmf, w  \models \Diamond \psi  &\text{ \, iff \, } \exists v \in W,
		\ifshort
	w \sim v \colon \Mmf, v \models \psi,
	\fi
	\ifappendix
	w \lhd v \colon \Mmf, v \models \psi,
	\fi
 \end{align*}
 \fi
 together with the usual interpretation of Boolean operators. An
 $\MLALCOud$ formula $\p$ is \emph{satisfied in $\Mmf$} if there
 exists a world $w$ in $\Mmf$ such that $\Mmf, w \models \p$, and it
 is \emph{partial} (\emph{total}) \emph{satisfiable} if there is a
 partial (total) modal interpretation $\Mmf$ such that $\p$ is
 satisfied in $\Mmf$.

\ifappendix
\subsection{Temporal free description logics}

For the temporal DL language $\TLALCOud$, we build $\TLALCOud$
\emph{terms}, \emph{concepts}, and \emph{formulas} similarly to the
$\MLALCOud$ case, by using the temporal operator \emph{until},
$\Until$, for the construction of concepts, $C \Until D$, and
formulas, $\p \Until \psi$.

We call \emph{flow of time} a frame $\Fmf = (T, <)$, where $T$ is a
non-empty set of \emph{instants} and $< \ \subseteq T \times T$ is a
strict linear order on $T$.  A \emph{partial temporal interpretation},
or \emph{partial trace}, based on a flow of time $\Fmf$, is a triple
$\Mmf = ( \Fmf, \Delta, \Int)$, defined as in the modal case.  We
similarly define the notion of \emph{total trace}.  If $\Mmf$ is based
on the flow of time $(\mathbb{N}, <)$, where $<$ is the natural strict
linear order on $\mathbb{N}$, we call it an \emph{infinite trace}, and
we often denote it (with an abuse of notation) by
$\Mmf = (\Delta, (\Imc_{t})_{t \in \mathbb{N}})$; whereas, if it is
based on $(\{0, \ldots, n \}, <)$, with $n \in \mathbb{N}$, it is
called a \emph{finite trace}, and it is simply denoted by
$\Mmf = (\Delta, (\Imc_{t})_{t \in T})$, with $T = [0, n]$.  \fi
%
\ifappendix
Given a partial trace $\Mmf = ( \Fmf, \Delta, \Int)$, with
$\Fmf = (T, <)$ and $t \in T$ (that we call an \emph{instant of
  $\Mmf$}), the \emph{value} of an $\TLALCOud$ term $\tau$ at $t$, and
the \emph{extension} of an $\TLALCOud$ concept $C$ at $t$, are defined
as for the modal case, by replacing the semantics of the $\Diamond$
modal operator with the following one for the $\Until$ temporal
operator:
\begin{gather*}
  (C \Until D)^{\Imc_{t}} = \{ d \in \Delta \mid \text{there is} \
			u \in T, t < u \colon d \in D^{\Imc_u} 
			\ \text{and, for all} \
			v \in (t,u)
			,
			d \in C^{\Imc_v}
		\}
			.
\end{gather*}
Similarly, the \emph{satisfaction} of a $\TLALCOud$ formula $\p$ at
$t$ of $\Mmf$ is defined as for the modal case, with the following
semantics of $\Until$ replacing the one for the 
$\Diamond$ operator:
\begin{gather*}
  \Mmf, t \models \p \Until \psi \ \text{iff} \ \text{there is} \
	u \in T, t < u \colon \Mmf, u \models \psi
	\ \text{and, for all $v \in (t,u)$,}\
	\Mmf, v \models \p.
\end{gather*}
\fi
%

\ifappendix
As usual, we use the until operator to define the other temporal
operators, as follows.
For concepts: \emph{(strong) next} operator,
$\Next C = \bot \Until C$; \emph{diamond} operator,
$\Diamond C = \top \Until C$; and \emph{box} operator,
$\Box C = \lnot \Diamond \lnot C$; \emph{reflexive diamond} operator,
$\Diamond^{+} C = C \sqcup \Diamond C$; \emph{reflexive box} operator,
$\Box^{+} C = C \sqcap \Box C$. Similar abbreviations are used for
formulas.

We say that a $\TLALCOud$ formula $\p$ (respectively, a concept $C$)
is \emph{satisfiable on partial (total) traces}, if $\p$
(respectively, $C$) is satisfied at time $0$ in some partial (total)
trace $\Mmf$.
%
\fi

\ifappendix
\begin{remark}
  Assertions are syntactic sugar, since $C(\tau)$ and
  $r(\tau_1,\tau_2)$ are captured by the following CIs, respectively:
  $\top \sqsubseteq \exists u. \{ \tau \}, \{ \tau \} \sqsubseteq C$;
  and
  $\top \sqsubseteq \exists u. \{ \tau_{1} \}, \{ \tau_{1} \}
  \sqsubseteq \exists r.\{ \tau_{2} \}$.  To avoid ambiguities, we use
  parentheses when applying Boolean or modal operators to
  assertions. Thus, for instance, the formulas $\lnot (C(\tau))$ and
  $\Diamond (C(\tau))$ abbreviate, respectively,
  $\lnot ( \top \sqsubseteq \exists u. \{ \tau \} \land \{ \tau \}
  \sqsubseteq C)$ and
  $\Diamond( \top \sqsubseteq \exists u. \{ \tau \} \land \{ \tau \}
  \sqsubseteq C)$, whereas the assertions $\lnot C(\tau)$ and
  $\Diamond C(\tau)$ stand, respectively, for
  $\top \sqsubseteq \exists u. \{ \tau \} \land \{ \tau \} \sqsubseteq
  \lnot C$ and
  $\top \sqsubseteq \exists u. \{ \tau \} \land \{ \tau \} \sqsubseteq
  \Diamond C$.
   Finally, as already observed for $\ALCOud$~\cite{ArtEtAl21a},
   we point out that formulas are just syntactic sugar in $\MLALCOud$,
   since a CI $C\sqsubseteq D$ can be internalised~\cite{BaaEtAl03a,Rud11} as a concept of
   the form $\forall u. ( C \Rightarrow D)$.
 \end{remark}

 \begin{remark}
We point out a counter-intuitive behaviour of formulas
without the RDA
assumption. Let us consider the following:
$(\{a\} \sqsubseteq \Box C) \land \Diamond (\{a\} \sqsubseteq \lnot
C)$. This formula, while unsatisfiable if the RDA is assumed, is
satisfiable
without the RDA,
since it is
satisfied
in
\ifshort
an epistemic or temporal
\fi
\ifappendix
a modal
\fi
interpretation that interprets the individual name $a$ differently in
different states.
\end{remark}

\fi

\subsection{Formula satisfiability problems and reductions}

Given a class of frames $\Cmc$, the \emph{partial} (\emph{total})
\emph{$\MLALCOud$ formula satisfiability problem over $\Cmc$}
(\emph{with} or \emph{without the RDA}, respectively) is the problem
of deciding, given an $\MLALCOud$ formula $\p$, whether there exists a
partial (total) modal interpretation (with or without the RDA,
respectively) based on a frame in $\Cmc$ that satisfies $\p$.
Similarly, for a class of flows of time $\Cmc$, the \emph{partial}
(\emph{total}) \emph{$\TLALCOud$ formula satisfiability problem over
  $\Cmc$} (\emph{with} or \emph{without the RDA}, respectively) is the
problem of deciding, given an $\TLALCOud$ formula $\p$, whether there
exists a partial (total) trace (with or without the RDA, respectively)
based on a flow of time in $\Cmc$ that satisfies $\p$.

The \emph{partial} (\emph{total}) \emph{$\SfiveALCOud$ formula
  satisfiability problem} (\emph{with} or \emph{without the RDA},
respectively) is the {partial} ({total}) $\MLALCOud$ formula
satisfiability problem (with or without the RDA, respectively) over
the class of \emph{epistemic frames} $(W, \sim)$ such that $\sim$ is
an \emph{equivalence relation} on $W$.

The \emph{partial} (\emph{total}) \emph{$\LTLALCOud$ formula
  satisfiability problem} (\emph{with} or \emph{without the RDA},
respectively) is the {partial} ({total}) $\TLALCOud$ formula
satisfiability problem (with or without the RDA, respectively) over
$\{ (\mathbb{N}, <) \}$.  The \emph{partial} (\emph{total})
\emph{$\LTLfALCOud$ formula satisfiability problem} (\emph{with} or
\emph{without the RDA}, respectively) is the {partial} ({total})
$\TLALCOud$ formula satisfiability problem (with or without the RDA,
respectively) over the class of finite strict linear orders of the
form $(\{ 0, \ldots, n \}, <)$, where $n \in \mathbb{N}$.

We first illustrate a polynomial-time reduction of formula
satisfiability on total modal or temporal interpretations without the
RDA, to the same problem over partial ones.
\begin{lemma}
\label{lemma:redtotaltopartial}
Total $\MLALCOud$ and $\TLALCOud$ formula satisfiability without the
RDA are polynomial-time reducible to, respectively, partial
$\MLALCOud$ and $\TLALCOud$ formula satisfiability without the RDA.
\end{lemma}
\begin{proof}
  We adapt, to the $\MLALCOud$ and $\TLALCOud$ cases, the reduction of
  $\ALCOud$ ontology satisfiability and entailment from total to
  partial interpretation, given in~\cite{ArtEtAl21a}.
  In particular, given a $\MLALCOud$ or $\TLALCOud$ formula $\p$, we
  define $\p'$ as the conjunction of $\p$ with formulas of the form
\begin{equation}
\label{eq:designate}
\tag{$\ast$}
\Box^{+} ( \top \sqsubseteq \exists u.\{a\}),
\end{equation}
for every individual name $a$ occurring in $\p$.  It can be seen that
$\p$ is satisfiable on \emph{total} modal or temporal interpretations
without the RDA iff $\p'$ is satisfiable on \emph{partial} modal or
temporal interpretations without the RDA, respectively.  Observe that,
despite the fact that each individual name $a$ occurring in $\p$ is
forced by~(\ref{eq:designate}) to denote at every world or instant of
any model of $\p'$, its interpretation is allowed to vary across
states, due to the lack of the RDA.
\end{proof}

In addition, we show how to remove, in polynomial time, definite
descriptions of the form $\iota C$, hence reducing the original
problem to the formula satisfiability problem in $\MLALCOu$.
\begin{lemma}
\label{lemma:redmludtomlu}
Partial $\MLALCOud$ and $\TLALCOud$ formula satisfiability without the
RDA are polynomial-time reducible to, respectively, partial $\MLALCOu$
and $\TLALCOu$ formula satisfiability without the RDA.
\end{lemma}
\begin{proof} Similar to the proof of~\cite[Lemma~1]{ArtEtAl21a},
  which shows that $\ALCOud$ formula satisfiability problem on partial
  interpretations is polynomial-time reducible to the $\ALCOu$ formula
  satisfiability problem on partial interpretations.
\end{proof}

Finally, given a formula $\p$ in $\MLALCOu$ or $\TLALCOu$, checking
for its satisfiability reduces to check concept satisfiability in
$\MLALCOu$ or $\TLALCOu$, respectively.
\begin{lemma}
\label{lemma:redformtoconcuniv}
Partial $\MLALCOu$ and $\TLALCOu$ formula satisfiability are
linear-time reducible to, respectively, partial $\MLALCOu$ and
$\TLALCOu$ concept satisfiability.
\end{lemma}
\begin{proof}
  This is an adaptation to the modal and temporal case of the
  so-called \emph{internalisation}
  property~\cite[Chapter~5]{BaaEtAl03a} of those DLs equipped with the
  universal role, allowing us to internalise a CI of the form
  $C\sqsubseteq D$ as a concept of the form
  $\forall u.(\lnot C \sqcup D)$.
\end{proof}


%% file: 3-motivations.tex
\section{Motivations and examples}
\label{sec:motivation}
In the following, we provide some motivating examples to illustrate
the novel features of the modal and temporal free DL
languages introduced in the previous section.

\paragraph{Epistemic scenario.}
We discuss here the main features of the language $\MLALCOud$
interpreted on partial modal interpretations without the RDA and based
on frames
equipped with an
equivalence relation,
to provide motivating examples in an epistemic
setting. Thus, in the following, we will refer to such interpretations
as \emph{epistemic contexts}, and the worlds of the epistemic context
will be called \emph{states}.

Let us consider the following individual names: $\mathsf{clark}$ standing
for the `Clark Kent', $\mathsf{superman}$ standing for `Superman', and
$\mathsf{lois}$ standing for `Lois Lane'. In a given state of an epistemic
context, $\mathsf{clark}$ and $\mathsf{superman}$ can denote the same element
of the domain. However, since the epistemic context does not satisfy
the RDA, these individual names can also refer to different
individuals at different states. Moreover, the term
$\defdes ( \mathsf{Journalist} \sqcap \exists
\mathsf{worksWith}.\{\mathsf{lois}\})$ can be used to refer to the
journalist that works with Lois Lane, which can be understood as an
alternative way to describe Clark Kent.

The concept $\Box\{\mathsf{superman}\}$ captures those individuals that are
\emph{known to be} Superman, whereas $\Diamond\{\mathsf{superman}\}$ would
be the set of individuals \emph{suspectable of being} (i.e., for which
it is not known that they are not) Superman.
%
The following CI can be used to express that Clark Kent is known to be
Clark Kent: $\{ \mathsf{clark} \} \sqsubseteq \Box \{ \mathsf{clark} \}$,
whereas we can use
$\{ \mathsf{clark} \} \sqsubseteq \lnot \Box \{ \mathsf{superman} \}$ to
express that Clark Kent is not known to be Superman. Moreover, the
following CIs assert that Lois Lane is known to be Lois Lane, i.e.,
$\{ \mathsf{lois} \} \sqsubseteq \Box \{ \mathsf{lois} \} $, and she is
also known to love Superman, i.e.,
$\{ \mathsf{lois} \} \sqsubseteq \Box \exists
\mathsf{loves}.\{\mathsf{superman}\}$,
but (unfortunately for Clark Kent)
not the journalist that works with
her, i.e., $\{\mathsf{lois}\} \sqsubseteq \lnot \Box \exists \mathsf{loves}.\{
\defdes ( \mathsf{Journalist} \sqcap \exists
\mathsf{worksWith}.\{\mathsf{lois}\}\}$.



\paragraph{Temporal scenario.}
In the following, we exemplify and discuss the main features of the
language $\TLALCOud$ interpreted on partial traces without the RDA.

Let us consider the individual name $\mathsf{dl}$ which can serve as a
``variable''-name, that might change its value (and even remain
uninterpreted) at different instants, so to denote, e.g., the
Description Logic (DL) workshops that take place over the years.
Instead, the individual name $\mathsf{dl22}$ can be used as an
individual name with a fixed referent across time, denoting the DL
workshop that takes place in 2022 only. Finally, the definite
description $\defdes\exists \mathsf{isGCof}.\{\mathsf{dl}\}$ is
able to refer non-rigidly to the General Chair of the DL workshop.

The concept $\Diamond^{+} \{\mathsf{dl}\}$ can be used to refer to
those objects that will eventually be the DL workshop. When
interpreted on total traces with the RDA, the extension of this
concept would contain exactly one object, i.e., the (rigid) referent
of the individual name $\mathsf{dl}$. On the other hand, on partial
traces without the RDA, its extension might contain also zero, or more
than one, objects, namely, those individuals that are the values of
$\mathsf{dl}$ at different time points. In addition, the concept
$\Diamond^{+}\{\defdes\exists\mathsf{isGCof}.\{\mathsf{dl}\}\}$ can be
used to represent those individuals that will eventually be the
General Chair of the various DL workshops.
The concept
$\Diamond^{+} \exists \mathsf{hasPCM}.\Diamond^{+} \{\defdes \exists
\mathsf{isGCof}.\{\mathsf{dl}\}\}$ captures those objects that will
eventually have as PC Members some individuals that, at some point in
the future, will be the General Chair of the DL workshop.

Let $\mathsf{proc\mbox{-}dl22}$ stands for the proceedings of the DL22
workshop. To enforce the proceedings of the DL22 to behave as a rigid
designator, we can enforce the following CI,
$\Diamond^{+} \{\mathsf{proc\mbox{-}dl22}\} \sqsubseteq
\Box^+\{\mathsf{proc\mbox{-}dl22}\}$.
%
On the other hand, we can require $\mathsf{dl}$ to be a
\emph{flexible} designator by enforcing the following CI,
$\top \sqsubseteq \exists u.(\Diamond^{+}\{\mathsf{dl}\} \sqcap
\Diamond^{+} \lnot \{ \mathsf{dl}\})$.
%
Finally, we can enforce $\mathsf{dl22}$ to be an \emph{instantaneous}
designator by using the following CIs,
$\top \sqsubseteq \Diamond^{+}\exists u.\{\mathsf{dl22}\}$,
$\Diamond^{+} (\{\mathsf{dl22}\} \sqcap
\Diamond\{\mathsf{dl22}\})\sqsubseteq \bot$, meaning that
$\mathsf{dl22}$ is not gost and that no objects can be the denotation
of the individual name $\mathsf{dl22}$ at two distinct instants.
Notice that gost individuals are assumed to be (vacuously) rigid,
while flexible and instantaneous individuals should denote.

To state that the (current) General Chair of DL will be one of next
year's PC Members of DL we can use the assertion,
$(\Next \exists {\sf isPCMof}.\{{\sf dl}\})(\{\defdes\exists {\sf
  isGCof}.\{ {\sf dl }\}\})$.
%
On the other hand, to say that, next year, the (future) General Chair
of DL will be one of the PC Members of DL, the following formula
should be used instead,
$\Next (\top \sqsubseteq \exists u.\{\defdes \exists {\sf isGCof}.
\{{\sf dl}\}\} \land \{\defdes \exists {\sf isGCof}.\{ {\sf dl}\}\}
\sqsubseteq \exists{\sf isPCMof}.\{{\sf dl}\})$.
Moreover, one can express that, at some point, DL22 will be the DL
workshop, by writing the following assertion,
$\Diamond^{+} (\{\mathsf{dl}\}(\{\mathsf{dl22}\}))$. Finally, the
formula
$\Box^{+}(\{\defdes\exists{\sf isGCof}.\{{\sf dl}\}\}
\sqsubseteq \exists{\sf isPCMof}.\Next\{{\sf dl}\})$ asserts
that it will always be the case that the General Chair of DL is going
to be a DL's PC Member on the subsequent year.
        

%% file: 4-ml-reasoning.tex
\section{Reasoning in epistemic free description logics}
\label{sec-reasoning}

In this section, we show that the formula satisfiability problem for
the epistemic free DL $\SfiveALCOud$
is
decidable and $\NExpTime$-complete.
%
%
After reducing, thanks to Lemmas~\ref{lemma:redmludtomlu} and~\ref{lemma:redformtoconcuniv}, satisfiability of a formula $\p$ in $\SfiveALCOud$ to
the satisfiability of a concept $C_\p$ in $\SfiveALCOu$, we can show
how to decide the satisfiability of the latter by adapting the technique of
\emph{quasimodels}~\cite{GabEtAl03} to the case where individual names
can be left uninterpreted and do not respect the RDA.

Given an
$\MLALCOu$
concept $C_\p$, let
$\con{C_\p}$ be the closure under single negation of the set of
concepts occurring in $C_\p$.
A \emph{type} for $C_\p$ is a subset $\contp$ of $\con{C_\p}$
such that:
\begin{enumerate}
  [label=\textbf{C\arabic*},leftmargin=*,series=run]
\item
  $\neg C \in \contp$ iff $C\not \in \contp$, for all
  $\neg C \in \con{C_\p}$;
  \label{ct:neg}
\item
  $C \sqcap D \in \contp$ iff $C, D \in \contp$, for all
  $C \sqcap D \in \con{C_\p}$.
  \label{ct:con}
\end{enumerate}
Note that there are at most exponentially many types, i.e., there are
$2^{|\con{C_\p}|}$ types for $C_\p$.

A \emph{quasistate for $C_\p$} is a non-empty set $\settp$ of types
for $C_\p$ satisfying the following conditions:
\begin{enumerate}
  [label=\textbf{Q\arabic*},leftmargin=*,series=run]
%
\item for every $\{a\} \in \con{C_\p}$, there exists at most one
  $\contp \in \settp$ such that $\{ a \} \in \contp$;
  \label{qs:nom}
\item for every $\contp \in \settp$ and every
  $\exists r.C \in \contp$, there exists $\contp' \in \settp$ such
  that $\{ \lnot D \mid \lnot \exists r. D \in \contp \} \cup \{ C \}
  \subseteq \contp'$;
  \label{qs:existsr}
\item for every $\contp \in \settp$, $\exists u.C \in \contp$ iff
  there exists $\contp' \in \settp$ such that $C \in \contp'$.
  \label{qs:univ}
\end{enumerate}

A \emph{basic structure for $C_\p$} is a pair $(W, \funcand)$, where
$W$ is a non-empty set, and $\funcand$ is a function associating with
every $w \in W$ a quasistate $\funcand(w)$ for $C_\p$, satisfying the
following condition:
\begin{enumerate}
  [label=\textbf{B\arabic*},leftmargin=*,series=run]
  \item there exists a world $w'\in W$ and a type $\contp \in
    \funcand(w')$ such that $C_\p \in \contp$.
    \label{b1}
\end{enumerate}
  
A \emph{run through $(W, \funcand)$} is a function $\rho$ mapping each
world $w \in W$ into a type $\rho(w) \in \funcand(w)$ and satisfying
the following condition for every $\Diamond C \in \con{\p}$:
\begin{enumerate}
[label=\textbf{R\arabic*},leftmargin=*,series=run]
\item $\Diamond C \in \rho(w)$ iff there exists $v \in W$ such that
  $C \in \rho(v)$.
  \label{rn:modal}
\end{enumerate}

An \emph{$\SfiveALCOu$ quasimodel for $C_\p$} is a triple $\quasimod =
(W, \funcand, \runs)$, where $(W, \funcand)$ is a basic structure for
$C_\p$, and $\runs$ is a set of runs through $(W, \funcand)$ such that
the following condition holds:
\begin{enumerate}
[label=\textbf{M\arabic*},leftmargin=*,series=run]
\item for every $w \in W$ and every $\contp \in \funcand(w)$, there
  exists $\rho \in \runs$ with $\rho(w) = \contp$;
  \label{run:exists}
\item for every $w\in W$ and for every
  $\contp\in \funcand(w), \text{ with } \{a\}\in \contp$, there exists
  exactly one $\rho\in \runs$ such that $\{a\} \in \rho(w)$.
    \label{run:nominal}
\end{enumerate}

We are now able to show how a quasimodel is related to the notion of
satisfiability.
\begin{proposition}\label{th-s5sat}
A
  concept $C_\p$ is
  partial
  $\SfiveALCOu$
  satisfiable
   without the RDA
  iff there exists an $\SfiveALCOu$ quasimodel for $C_\p$.
\end{proposition}
\ifappendix
\begin{proof}
$(\Rightarrow)$
Let $\Mmf = (\Fmf, \Delta, \Imc)$, with
  $\Fmf = (W, \sim)$,
  be a partial modal interpretation
  satisfying $C_\p$.
 Without loss of generality~\cite{GabEtAl03}, we can assume that $\sim \ = W \times W$.
  Consider
  $\contp^{\Imc_{w}}(d) = \{ C \in \con{C_\p} \mid d \in
  C^{\Imc_{w}}\}$, for every $d \in \Delta$ and $w \in W$. Clearly,
  $\contp^{\Imc_{w}}(d)$ is a concept type for $C_\p$ since it clearly
  satisfies \ref{ct:neg}-\ref{ct:con}. We now define a
  triple $\Qmf = (W, \funcand, \runs)$, where:
  \begin{itemize}
  \item $\funcand$ is a function from $W$ to the set of
    quasistates for $C_\p$ such that
    $\funcand(w) = \{ \contp^{\Imc_{w}}(d) \mid d \in \Delta\}$, for
    every $w\in W$;
  \item $\runs$ is the set of functions $\rho_{d}$ from $W$ to
    the set of types for $C_\p$ such that
    $\rho_{d}(w) = \contp^{\Imc_{w}}(d)$, for every
    $d \in \Delta$ and for every $w\in W$.
   \end{itemize}
   It is easy to show that $\Qmf$ is a quasimodel for $C_\p$. Indeed,
   $\funcand$ is well-defined, as $\funcand(w)$ is a set of types for
   $C_\p$ satisfying~\ref{qs:nom}-\ref{qs:univ}, for every $w \in
   W$. Moreover, $(W, \funcand)$ is a basic structure for $C_\p$ since
   $\Mmf$ is satisfying $C_\p$, i.e., $\Mmf,w\models C_\p(d)$, for
   some $w\in W$ and $d\in\Delta$, and thus $(W, \funcand)$
   satisfies~\ref{b1}. The set of runs, $\runs$, by construction,
   satisfies~\ref{rn:modal}. Finally, by definition of $\funcand(w)$
   and of $\rho$, $\Qmf$ satisfies~\ref{run:exists}, and, since for
   every $w\in W$, if $\{a\}\in \contp$, then there should be exactly
   one $d\in\Delta$ such that $d=\{a\}^{\Imc_{w}}$, thus $\Qmf$
   satisfies~\ref{run:nominal}.

   \noindent $(\Leftarrow)$ Suppose there is a quasimodel
   $\Qmf = (W, \funcand, \runs)$ for $C_\p$. Define a partial modal
   interpretation $\Mmf = (\Fmf, \Delta, \Imc)$, with
   $\Fmf = (W, W \times W)$, $\Delta = \runs$, and, for any $A\in\NC$,
   $r\in\NR$ and $a\in\NI$ the following holds:
   \begin{itemize}
   \item $A^{\Imc_{w}} = \{ \rho \in \Delta \mid A \in \rho(w) \}$;
   \item
     $r^{\Imc_{w}} = \{ (\rho, \rho') \in \Delta \times \Delta \mid \{
     \lnot C \mid \lnot \exists r . C \in \rho(w) \} \subseteq
     \rho'(w) \}$;
   \item $u^{\Imc_{w}} = \Delta \times \Delta$;
   \item $a^{\Imc_{w}} = \rho$, for the (unique, if any)
     $\rho \in \runs$ such that $\{ a \} \in \rho(w)$ (undefined,
     otherwise).
   \end{itemize}
   Observe that $\Delta$ is well-defined since $W$ is a non-empty set
   and $\funcand(w)\neq\emptyset$, for all $w\in W$. Thus,
   by~\ref{run:exists}, $\runs\neq\emptyset$. Also,
   by~\ref{qs:nom} and~\ref{run:nominal}, $a^{\Imc_{w}}$ is
   well-defined. We now require the following claim.
   \begin{claim}\label{cla:qmcon}
     For every $C \in \con{C_\p}$, $w \in W$ and $\rho \in \Delta$,
     $\rho \in C^{\Imc_{w}}$ iff $C \in \rho(w)$.
   \end{claim}
   \begin{proof}
     The proof is by induction on $C$. The base cases, $C = A$ and $C =
     \{a\}$, follow immediately from the definition of $\Mmf$. We then
     consider the inductive cases.\\
     Let $C = \lnot D$. $\lnot D\in \rho(w)$ iff, by~\ref{ct:neg},
     $D\not\in \rho(w)$. By induction, $D\not\in \rho(w)$ iff
     $\rho\not\in D^{\Imc_{w}}$ iff $\rho\in(\lnot D)^{\Imc_{w}}$.\\
     Let $C = D\sqcap E$. Similar to the previous case, now by
     using~\ref{ct:con}.\\
     Let $C = \exists u.D$. $\rho \in (\exists u.D)^{\Imc_{w}}$ iff
     there exists $\rho' \in D^{\Imc_{w}}$. By inductive hypothesis,
     $\rho' \in D^{\Imc_{w}}$ iff $D \in \rho'(w)$. By~\ref{qs:univ},
     the previous step holds iff $\exists u.D \in \rho(w)$.\\
     Let $C = \exists r.D$. $(\Rightarrow)$ Suppose that
     $\rho\in (\exists r.D)^{\Imc_{w}}$. Then, there exists
     $\rho' \in \Delta$ such that $(\rho, \rho') \in r^{\Imc_{w}}$ and
     $\rho' \in D^{\Imc_{w}}$. By inductive hypothesis,
     $\rho' \in D^{\Imc_{w}}$ iff $D\in\rho'(w)$. By contradiction,
     assume that $\exists r.D\not\in\rho(w)$, then, by~\ref{ct:neg},
     $\lnot \exists r.D\in\rho(w)$. By definition of $r^{\Imc_{w}}$,
     since $(\rho,\rho')\in r^{\Imc_{w}} $, then,
     $\lnot D\in \rho'(w)$, thus contadicting, by~\ref{ct:neg}, that
     $D\in \rho'(w)$.
     $(\Leftarrow)$ Conversely, suppose that
     $\exists r.D \in \rho(w)$. By~\ref{qs:existsr}
     and~\ref{run:exists}, there exists a $\rho' \in \Delta$ such that
     $\{ \lnot E \mid \lnot \exists r . E \in \rho(w) \} \cup \{ D \}
     \subseteq \rho'(w)$. By inductive hypothesis and the definition
     of $r^{\Imc_{w}}$, $\rho' \in  D^{\Imc_{w}}$ and $(\rho, \rho')
     \in r^{\Imc_{w}}$. Thus, $\rho \in (\exists r.D)^{\Imc_{w}}$.\\
     Let $C = \Diamond D$.  $\rho \in (\Diamond D)^{\Imc_{w}}$ iff
     there exists $v \in W$ s.t. $\rho\in D ^{\Imc_{v}}$.  By
     inductive hypothesis, $\rho \in D^{\Imc_{v}}$ iff $D \in \rho(v)$
     iff, by~\ref{rn:modal}, $\Diamond D\in\rho(w)$.
\end{proof}
Now we can easily finish the proof of the proposition by observing
that, by~\ref{b1}, there exists a world $w'\in W$ and a type
$\contp \in \funcand(w')$ such that $C_\p \in \contp$. Thus,
by~\ref{run:exists}, there exists $\rho \in \runs$ with
$\rho(w') = \contp$ and, by the above claim,
$\rho\in C_{\p}^{\Imc_{w'}}$.
\end{proof}

We now show that if there exists a quasimodel it exists one of
exponential size.
%
\begin{theorem}
  There exists an $\SfiveALCOu$ quasimodel for $C_\p$ iff there exists
  an $\SfiveALCOu$ quasimodel for $C_\p$ of exponential size in the
  length of $C_\p$.
\end{theorem}
\begin{proof}
  The $(\Rightarrow)$ direction is straightforward.
  For the $(\Leftarrow)$ direction, we adapt the proof
  of~\cite[Theorem 5.25]{GabEtAl03} where the notion of \emph{twin} is
  substituted by a new notion of \emph{w-copies}. Intuitively, the
  notion of \emph{twin} does not work here due to the absence of the
  RDA. Indeed, the following $\SfiveALCOu$ concept:
  $$
  \exists u.(C \sqcap \Diamond \{a\}) \sqcap \exists u.(\lnot C \sqcap \Diamond \{a\})
  $$
  is satisfiable only without the RDA thanks to the possibility to
  interpret the nominal concept $\{a\}$ as two different domain
  individuals at two different worlds.
  Let's now show the construction of the exponentially bounded
  quasimodel based on the notion of \emph{w-copies}. Suppose that
  $\Qmf = (W, \funcand, \runs)$ is a quasimodel for $C_\p$.  Now, we
  construct from $\Qmf$ a quasimodel
  $\Qmf' = (W_1\cup W_2, \funcand', \runs')$, where
  $\funcand'$ is the restriction of $\funcand$ to $W_1\cup W_2$.
  We start with the construction of $W_1$ and a set of runs
  $\gruns$. First, we add to $W_1$ a world $w_{\p}\in W$ such that
  $C_{\p} \in \contp$, for some $\contp \in \funcand(w_{\p})$. Such a
  world exists due to~\ref{b1}. We then associate with every
  $\contp \in \funcand(w_{\p})$ a run $\rho_{\contp}$ such that
  $\rho_{\contp}(w_{\p}) = \contp$. These runs exist due
  to~\ref{run:exists}. For each $\contp\in \funcand(w_{\p})$ and for
  each $\Diamond C\in \contp$, select a world $w\in W$ such that
  $C\in\rho_{\contp}(w)$ (such a $w$ exists by~\ref{rn:modal}), and
  add $w$ to $W_1$. Thus, the resulting $W_1\subseteq W$ contains at most
  $2^{|\con{C_\p}|}\cdot {|\con{C_\p}|}$ worlds.
  Denote by $\rho'_{\contp}$ the restriction of $\rho_{\contp}$ to
  $W_1$, for every $\contp \in \funcand(w_{\p})$.  By construction,
  for every $\contp \in \funcand(w_{\p})$, $\rho'_{\contp}$ is a run
  through $(W_1, \funcand')$.  We thus set
  $\gruns = \{ \rho'_{\contp} \mid \contp \in \funcand(w_{\p}) \}$, and
  add all runs in $\gruns$ to $\runs'$.
  We now proceed to the construction of $W_2$ and $\runs'$ by using
  \emph{w-copies} of worlds in $W_1$, and by either introducing new
  runs or extending runs in $\gruns$ to worlds in $W_2$. This new
  construction is needed to satisfy property~\ref{run:exists} of
  quasimodels. Indeed, $\gruns$ may not contain runs coming through
  all types in $(W_1, \funcand')$.
  For every $w\in W_1$, let $\contp\in \funcand'(w)$ such that there
  is no $\rho \in \gruns$ with $\contp\in\rho(w)$. Then, there exists
  $t'\in \funcand'(w)$, a run $\rho\in\gruns$, and a run
  $\rho_{\contp,w}\in\runs$ such that:
  \begin{enumerate}
    [label=\textbf{T\arabic*},leftmargin=*,series=run]
  \item $\rho(w) = \contp'$ and $\rho_{\contp,w}(w) = \contp$;
   \item $\rho(w_\p) = \rho_{\contp,w}(w_\p)$;
   \item for any $\Diamond C\in \contp$, $\Diamond C\in \contp'$;
     \label{t3}
   \item if $\Diamond C, \lnot C\in \contp$, then, $\Box C\not\in
     \contp'$.
     \label{t4}
  \end{enumerate}
  The above properties are guaranteed by $\Qmf = (W, \funcand, \runs)$
  being a quasimodel for $C_\p$.
  We thus add $\rho_{\contp,w}$ to $\runs'$ and construct $W_2$ using
  \emph{w-copies} of quasistates in $W_1$ in such a way that: for
  every $w_2\in W_2$ there is a unique $w_1 \in W_1$ such that $w_2$
  is a \emph{w-copy} of $w_1$, and $\funcand'(w_2) =
  \funcand'(w_1)$. Such \emph{w-copies} are added to satisfy
  sub-concepts of the form $\Diamond C$ possibly present in
  $\contp$. We distinguish the following two cases:
  \begin{itemize}
  \item Let $\Diamond \{a\},\lnot \{a\} \in \contp$. Then, let
    $w_a\in W_1$ such that $\{a\}\in\rho(w_a)$ (wich exists
    by~\ref{t3} and by construction of $W_1$), add a fresh new
    \emph{w-copy} of $w_a$, say $w'_a$, to $W_2$ and extend the two
    runs, $\rho$ and $\rho_{\contp,w}$, in the following way:
    \begin{align*}
    \rho_{\contp,w}(w'_a) = \rho(w_a);\quad\quad
    \rho(w'_a) = \rho_{\contp,w}(w_a).
    \end{align*}
    For all other runs in $r\in\runs'$ their type is $w'_a$ is the
    same as the one in $w_a$, i.e., $r(w'_a) = r(w_a)$.
  \item Let $\Diamond C,\lnot C \in \contp$. Then, let $w_c\in W_1$
    such that $C\in\rho(w_a)$ (wich exists by~\ref{t3} and by
    construction of $W_1$) and $C\not\in \rho_{\contp,w}(w_c)$ (for
    otherwise we don't need to introduce a new \emph{w-copy}). Then,
    add a fresh new \emph{w-copy} of $w_c$, say $w'_c$, to $W_2$ and
    extend the two runs, $\rho$ and $\rho_{\contp,w}$, in the
    following way:
    \begin{align*}
    \rho_{\contp,w}(w'_c) = \rho(w_c);\quad\quad
    \rho(w'_c) = \rho_{\contp,w}(w_c).
    \end{align*}
    For all other runs in $r\in\runs'$ their type is $w'_c$ is the
    same as the one in $w_c$, i.e., $r(w'_c) = r(w_c)$.
  \end{itemize}
  It is easy to show that all elements in $\runs'$ are indeed runs
  mainly due to~\ref{t4} and the way \emph{w-copies} are added in
  $W_2$. At the same time, we notice that $\Qmf'$
  respects~\ref{run:exists} and, in particular, \ref{run:nominal} due
  to the way we extend runs.  As for the size of the set of worlds in
  $\Qmf'$, we have that $|W_1\cup W_2|$ contains at most
  $2^{2|\con{C_\p}|}\cdot {|\con{C_\p}|^2}$ worlds.
\end{proof}
\fi

By the previous theorem, we obtain decidability of the $\SfiveALCOu$
formula satisfiability problem in $\NExpTime$. Indeed, to check
whether an $\SfiveALCOu$ formula $\p$ is satisfiable, one can guess a
triple $\Qmf = (W, \funcand, \runs)$ of size exponential in the length
of $\p$, and then check whether $\Qmf$ is an $\SfiveALCOu$ quasimodel
for $\p$.
%
Given that already $\SfiveALC$ formula satisfiability (without nominals and universal role) is known to be $\NExpTime$-hard~\cite[Theorem14.14]{GabEtAl03}, we obtain a matching lower bound.
By Lemma~\ref{lemma:redmludtomlu}, we also get the following result.

\begin{theorem}
Partial $\SfiveALCOud$ formula satisfiability without the RDA is $\NExpTime$-complete.
\end{theorem}


%% file: 5-time-reasoning.tex
\section{Reasoning in temporal free description logics}
\label{sec:reasontfdl}


In the following,\nb{A: added sentence} we show that temporal DLs
without RDA are undecidable, differently from the case with RDA which
has been shown to be decidable in the absence of definite
descriptions~\cite[Theorem 14.12]{GabEtAl03}.
We start considering the DL $\LTLfALCOu$, obtained from $\LTLfALCOud$
by dropping definite descriptions.  We show that, without the RDA for
individual names, the formula satisfiability problem is
undecidable. This\nb{A: added} result holds already for \emph{total}
interpretations (hence also for the
\emph{partial} case).
The undecidability is due to the interaction between temporal
operators, the universal role, and nominals interpreted non-rigidly
over time.
The main proof, that adapts an analogous one due to Degtyarev et
al.~\cite{DegEtAl02}, goes by a reduction of the halting problem for
Minsky machines to total $\LTLfALCOu$ formula satisfiability without
the RDA.

\begin{theorem}
\label{prop:tlalcoufintotundec}
$\LTLfALCOu$\nb{A: removed ``Total''} formula satisfiability without
the RDA is undecidable.
\end{theorem}
\begin{proof}
  We\nb{A: added sentence} consider the simpler case of a total
  interpretation. The result is obtained by a reduction of the
  (undecidable) halting problem for (two-counter) Minsky machines
  starting with $0$ as initial counters
  values~\cite{DegEtAl02,BaaEtAl17}, that can be encoded
in
$\LTLfALCOu$
without
the RDA.

A \emph{(two-counter) Minsky machine}
is a pair $M = (Q, P)$, where $Q = \{ q_{0}, \ldots, q_{L} \}$ is a set of \emph{states} and $P = ( I_{0}, \ldots, I_{L - 1} )$ is a sequence of \emph{instructions}.
We assume $q_{0}$ to be the \emph{initial state}, and $q_{L}$ to be the \emph{halting state}.
Moreover, the instruction $I_{i}$ is executed at state $q_{i}$, for $0 \leq i < L$.
Each instruction $I$ can have one of the following forms, where $r \in \{ r_{1}, r_{2} \}$ is a \emph{register}, that stores non-negative integers as values,
and $q,q'$ are states:
\begin{enumerate}
	\item $I = +(r,q)$ -- add $1$ to the value of register $r$ and go to state $q$;
	\item $I = -(r,q,q')$ -- if the value of register $r$ is (strictly) greater than $0$, subtract $1$ to the value of $r$ and go to state $q$; otherwise, go to state $q'$.
\end{enumerate}
Given a Minsky machine $M$,
a \emph{configuration of $M$} is a triple $(q, v_{r_{1}}, v_{r_{2}})$, where $q$ is a state of $M$ and $v_{r_{1}}, v_{r_{2}} \in \mathbb{N}$ are the \emph{values} of registers $r_{1}, r_{2}$, respectively.
In the following, we set
$\overline{r_{k}} = r_{3 - k}$, for $k \in \{ 1, 2 \}$.
Given $i, j \geq 0$, we write $(q_{i}, v_{r_{1}}, v_{r_{2}}) \Rightarrow_{M} (q_{j}, v'_{r_{1}}, v'_{r_{2}})$ iff one of the following conditions hold:
\begin{itemize}
	\item $I_{i} = +(r, q_{j})$, $v'_{r} = v_{r} + 1$ and $v'_{\overline{r}} = v_{\overline{r}}$;
	\item $I_{i} = -(r, q_{j}, q')$, $v_{r} > 0$, $v'_{r} = v_{r} - 1$, and $v'_{\overline{r}} = v_{\overline{r}}$;
	\item $I_{i} = -(r, q, q_{j})$, $v_{r} = v'_{r} = 0$,
	and
	$v'_{\overline{r}} = v_{\overline{r}}$.
\end{itemize}
Given an \emph{input} $(u, t) \in \mathbb{N} \times \mathbb{N}$, the \emph{computation of $M$ on input $(u,t)$} is the (unique) longest sequence of configurations
$(p_{0}, v^{0}_{r_{1}}, v^{0}_{r_{2}}) \Rightarrow_{M} (p_{1}, v^{1}_{r_{1}}, v^{1}_{r_{2}}) \Rightarrow_{M} \ldots$,
such that $p_{0} = q_{0}$, $v^{0}_{r_{1}} = u$ and $v^{0}_{r_{2}} = t$.
We say that $M$ \emph{halts} \emph{on input} $(0,0)$ if the computation of $M$ on input $(0,0)$ is finite: thus, its initial configuration takes the form $(q_{0}, 0, 0)$, while its last configuration has the form $(q_{L}, v_{r_{1}}, v_{r_{2}})$.
%
The \emph{halting problem for Minsky machines} is the problem of deciding, given a Minsky machine $M$, whether $M$ halts on input $(0,0)$.
%
This problem is known to be undecidable~\cite{DegEtAl02,BaaEtAl17}.

To represent the computation of a Minsky machine $M$, we use the temporal dimension to model successive configurations in the computation of $M$ on input $(0,0)$. We introduce a concept name $Q_{i}$, for each state $q_{i}$ of $M$. Concept names $R_{1}, R_{2}$ are used to represent the registers $r_{1}, r_{2}$, respectively: the cardinality of their extensions at a given instant
will capture the values
of the registers
at one step
of the computation.
The incrementation, respectively, the decrementation, by one unit to the value of register $r_{k}$
is modelled by requiring the extensions of
concepts $\lnot R_{k} \sqcap \Next R_{k}$, respectively $R_{k} \sqcap \Next \lnot R_{k}$, to be included in those of nominals $\{ a_{r_{k}} \}$ and $\{ b_{r_{k}} \}$, respectively.
Since the individual names $a_{r_{k}}$ and $b_{r_{k}}$ are interpreted \emph{non-rigidly},
the (unique) element added or subtracted from the extension of $R_{k}$ is forced to vary over time.
In the following, for $k \in \{ 1, 2 \}$, we set $\overline{R_{k}} = R_{3 - k}$.
%

The initial configuration $(q_{0}, 0, 0)$ of $M$, together with the assumptions that at each configuration the machine is in exactly one state, and that the halting state is reached only at the last configuration, are represented by the conjunction of the following formulas (recall that $\last := \Box \bot$):
\begin{enumerate}
[label=$(\textbf{S\arabic*})$, align=left, leftmargin=*,series=run]
	\item $R_{1} \sqsubseteq \bot$;\label{item:count1}
	\item $R_{2} \sqsubseteq \bot$;\label{item:count2}
	\item $\top \sqsubseteq Q_{0}$;\label{item:initstate}
	\item $\Box^{+} (\top \sqsubseteq \bigsqcup_{i = 1}^{L} Q_{i})$;\label{item:cover}
	\item $\bigwedge_{i = 1}^{L} \Box^{+} (Q_{i} \sqsubseteq \bigsqcap_{j \neq i} \lnot Q_{j})$;\label{item:disjoint}
	\item $\Box^{+} (Q_{L} \equiv \last)$.\label{item:last}
\end{enumerate}
Instructions of the form $I_{i} = +(r_{k}, q_{j})$ are represented by the conjunction of the following formulas:
\begin{enumerate}
[label=$(\textbf{I\arabic*})$, align=left, leftmargin=*,series=run]
	\item $\Box^{+} ( Q_{i} \sqsubseteq \exists u.(\lnot R_{k} \sqcap \Next R_{k}))$;\label{item:incatleastone}
	\item $\Box^{+} ( Q_{i} \sqsubseteq \forall u. (\lnot R_{k} \sqcap \Next R_{k} \Rightarrow  \{ a_{r_{k}} \}))$;\label{item:incatmostone}
	\item $\Box^{+} ( Q_{i} \sqsubseteq \forall u. (R_{k} \Rightarrow \Next R_{k}))$;\label{item:incprescount}
	\item $\Box^{+} ( Q_{i} \sqsubseteq \forall u. (\overline{R_{k}} \Leftrightarrow \Next \overline{R_{k}}))$;\label{item:incothercount}
	\item $\Box^{+} ( Q_{i} \sqsubseteq \Next Q_{j})$.\label{item:incnext}
\end{enumerate}
Instructions of the form $I_{i} = -(r_{k}, q_{j}, q_{h})$ are
represented by the conjunction of the following formulas:
\begin{enumerate}
[label=$(\textbf{D\arabic*})$, align=left, leftmargin=*,series=run]
	\item $\Box^{+} ( Q_{i}  \sqcap \exists u. R_{k} \sqsubseteq \exists u.( R_{k} \sqcap \Next \lnot R_{k}))$;\label{item:decatleastone}
	\item $\Box^{+} ( Q_{i}  \sqcap \exists u. R_{k} \sqsubseteq \forall u. (R_{k} \sqcap \Next \lnot R_{k} \Rightarrow \{ b_{r_{k}} \} ))$;\label{item:decatmostone}
	\item $\Box^{+} ( Q_{i} \sqcap \exists u. R_{k} \sqsubseteq \forall u. (\Next R_{k} \Rightarrow R_{k}))$;\label{item:decprescount}
	\item $\Box^{+} ( Q_{i} \sqcap \exists u. R_{k} \sqsubseteq \forall u. (\overline{R_{k}} \Leftrightarrow \Next \overline{R_{k}}))$;\label{item:decothercount}
		\item $\Box^{+} ( Q_{i} \sqcap \lnot \exists u. R_{k} \sqsubseteq \forall u. ({R_{1}} \Leftrightarrow \Next {R_{1}}) \sqcap \forall u. ({R_{2}} \Leftrightarrow \Next {R_{2}}) )$;\label{item:decelsecount}
	\item $\Box^{+} ( Q_{i} \sqcap \exists u. R_{k} \sqsubseteq \Next Q_{j})$;\label{item:decnext}
	\item $\Box^{+} ( Q_{i} \sqcap \lnot \exists u. R_{k} \sqsubseteq \Next Q_{h})$.\label{item:decelsenext}
\end{enumerate}

Let $\chi^{M}$ be the conjunction of the CIs above, for $0 \leq i < L$. We require the following claim.

\begin{claim}
\label{cla:minskyhalt}
A
Minsky machine
$M$ halts
on input $(0,0)$
iff
$\chi_{M}$
is satisfiable
on
finite
total
traces
without the RDA.
\end{claim}
\begin{proof}[Proof of Claim~\ref{cla:minskyhalt}]

$(\Rightarrow)$
Let $M$ be a Minsky machine that halts on input $(0,0)$, and let $(p_{0}, 0, 0) \Rightarrow_{M} \ldots \Rightarrow_{M} (p_{n}, v^{n}_{r_{1}}, v^{n}_{r_{2}})$ be the computation of $M$ on input $(0,0)$, where $p_{0} = q_{0}$, $n \geq 0$ and $p_{n} = q_{L}$.
We construct a finite
total
trace
$\Mmf = (\Delta, (\Imc_{t})_{t \in \Tmf})$
such that $\Mmf, 0 \models \chi_{M}$.
%
Let $\Tmf = [0, n]$,
and
let $\Delta$
be a
fixed
countable set.
Given
$t \in [ 0, n ]$
and $p_{t} = q_{i}$, for some
$0 \leq i \leq L$,
we set $Q_{i}^{\Imc_{t}} = \Delta$ and $Q_{j}^{\Imc_{t}} = \emptyset$, for all $j \neq i$.
%
In addition, we set $R_{1}^{\Imc_{0}} = R_{2}^{\Imc_{0}} = \emptyset$, while
$a_{r_{k}}^{\Imc_{0}} = d \in \Delta$, and $b_{r_{k}}^{\Imc_{0}} = e \in \Delta$, for some $d \neq e$.
Moreover, for an instant
$0 \leq t < n$,
we define the interpretations $R_{1}^{\Imc_{t+1}}$, $R_{2}^{\Imc_{t+1}}$, $a_{r_{k}}^{\Imc_{t+1}}$, and $b_{r_{k}}^{\Imc_{t+1}}$
inductively as follows.
\begin{itemize}
	\item If $p_{t} = q_{i}$ and $I_{i} = +(r_{k}, q_{j})$, then we set:
	$R_{k}^{\Imc_{t + 1}} = R_{k}^{\Imc_{t}} \cup \{ d\}$, where $d = a_{r_{k}}^{\Imc_{t}}$;
	$a_{r_{k}}^{\Imc_{t+1}} = d'$, for a $d' \not \in R_{k}^{\Imc_{t}}$;
	$b_{r_{k}}^{\Imc_{t+1}} = a_{r_{k}}^{\Imc_{t}}$;
	$\overline{R_{k}}^{\Imc_{t + 1}} = \overline{R_{k}}^{\Imc_{t}}$.
	\item If $p_{t} = q_{i}$ and $I_{i} = -(r_{k},q_{j},q_{h})$, then:
		\begin{itemize}
			\item if
			$R_{k}^{\Imc_{t}} \neq \emptyset$,
			then we set:
	$R_{k}^{\Imc_{t + 1}} = R_{k}^{\Imc_{t}} \setminus \{ e \}$, where $e = b_{r_{k}}^{\Imc_{t}}$;
	$b_{r_{k}}^{\Imc_{t+1}} = e'$,
	for an $e' \in R_{k}^{\Imc_{t + 1}}$, if $R_{k}^{\Imc_{t + 1}} \neq \emptyset$, and $b_{r_{k}}^{\Imc_{t+1}}$ arbitrary, otherwise;
	$a_{r_{k}}^{\Imc_{t+1}} = a_{r_{k}}^{\Imc_{t}}$;
	$\overline{R_{k}}^{\Imc_{t + 1}} = \overline{R_{k}}^{\Imc_{t}}$.
			\item if
			$R_{k}^{\Imc_{t}} = \emptyset$,
			then we set: $R_{1}^{\Imc_{t + 1}} = R_{1}^{\Imc_{t}}$, $R_{2}^{\Imc_{t + 1}} = R_{2}^{\Imc_{t}}$; $a_{r_{k}}^{\Imc_{t+1}} = a_{r_{k}}^{\Imc_{t}}$; 	$b_{r_{k}}^{\Imc_{t+1}} = b_{r_{k}}^{\Imc_{t}}$.
		\end{itemize}
\end{itemize}
It can be seen that, by construction, $\Mmf, 0 \models \chi_{M}$.

$(\Leftarrow)$ Given a Minsky machine $M$, suppose that $\chi_{M}$
is satisfied on a
finite
total
trace
$\Mmf = (\Delta, (\Imc_{t})_{t \in \Tmf})$, with $\Tmf = [0, n]$.
%
To each time point
$t \in [ 0, n ]$,
we associate a configuration of $M$ of the form $(p_{t}, v^{t}_{r_{1}}, v^{t}_{r_{2}})$ so that $p_{t} = q_{j}$, where $q_{j}$ is the state of $M$ such that $\Mmf, t \models \top \sqsubseteq Q_{j}$, which is entailed by~\ref{item:cover}~and~\ref{item:disjoint}, and $v^{t}_{r_{k}} = | R_{k}^{\Imc_{t}} |$.
%

Let $t = 0$.
By~\ref{item:initstate}, we have $p_{0} = q_{0}$,
and $v^{0}_{r_{k}} = | R_{k}^{\Imc_{0}} | = 0$,
by~\ref{item:count1},~\ref{item:count2}.
Thus, the initial configuration of the computation of $M$ is
$(q_{0}, 0, 0)$.

For $0 \leq t < n$, we show that $(p_{t}, v^{t}_{r_{1}}, v^{t}_{r_{2}}) \Rightarrow_{M} (p_{t + 1}, v^{t + 1}_{r_{1}}, v^{t + 1}_{r_{2}})$. Let $p_{t} = q_{i}$.

\begin{itemize}
	\item Suppose $I_{i} = +(r_{k}, q_{j})$. By~\ref{item:incnext}, we have that $\Mmf, t + 1 \models \top \sqsubseteq Q_{j}$, and thus $p_{t + 1} = q_{j}$. Moreover, by~\ref{item:incatleastone},~\ref{item:incatmostone}, and~\ref{item:incprescount}, we have that $| R_{k}^{\Imc_{t + 1}} | = | R_{k}^{\Imc_{t}} | + 1$, thus $v^{t + 1}_{r_{k}} = v^{t}_{r_{k}} + 1$. Instead, by~\ref{item:incothercount}, $| \overline{R_{k}}^{\Imc_{t + 1}} | = | \overline{R_{k}}^{\Imc_{t}} |$, hence $v^{t + 1}_{\overline{r_{k}}} = v^{t}_{\overline{r_{k}}}$.
	\item Suppose $I_{i} = -(r_{k}, q_{j}, q_{h})$.
		\begin{itemize}
			\item If $\Mmf, t \models \top \sqsubseteq \exists u. R_{k}$, we have by~\ref{item:decnext} that $\Mmf, t + 1 \models \top \sqsubseteq Q_{j}$, and thus $p_{t + 1} = q_{j}$. Moreover, by~\ref{item:decatleastone},~\ref{item:decatmostone}, and~\ref{item:decprescount}, we have that $| R_{k}^{\Imc_{t + 1}} | = | R_{k}^{\Imc_{t}} | - 1$, thus $v^{t + 1}_{r_{k}} = v^{t}_{r_{k}} - 1$. Instead, by~\ref{item:decothercount}, $| \overline{R_{k}}^{\Imc_{t + 1}} | = | \overline{R_{k}}^{\Imc_{t}} |$, hence $v^{t + 1}_{\overline{r_{k}}} = v^{t}_{\overline{r_{k}}}$.
			\item If $\Mmf, t \not \models \top \sqsubseteq \exists u. R_{k}$, we have by~\ref{item:decelsenext} that $\Mmf, t + 1 \models \top \sqsubseteq Q_{h}$, and thus $p_{t + 1} = q_{h}$. Moreover, by~\ref{item:decelsecount}, we have both that $| R_{k}^{\Imc_{t + 1}} | = | R_{k}^{\Imc_{t}} | = 0$, thus $v^{t + 1}_{r_{k}} = v^{t}_{r_{k}} = 0$, and $| \overline{R_{k}}^{\Imc_{t + 1}} | = | \overline{R_{k}}^{\Imc_{t}} |$, hence $v^{t + 1}_{\overline{r_{k}}} = v^{t}_{\overline{r_{k}}}$.
		\end{itemize}
\end{itemize}

Finally, since
$n$ is the last instant of the trace, i.e., $\Mmf, n \models \top \sqsubseteq \last$, we have by~\ref{item:last} that $\Mmf, n \models \top \sqsubseteq Q_{L}$,
Therefore, the last configuration of the computation of $M$ (and the last only) takes the form $(q_{L}, v^{n}_{r_{1}}, v^{n}_{r_{2}})$, as required.
\end{proof}

Thanks to the previous claim, and the fact that the halting problem for Minsky machines is undecidable~\cite{DegEtAl02,BaaEtAl17}, we conclude the proof.
\end{proof}


As a next step, we require the following polynomial-time reduction, from formula satisfiability on finite traces to the same problem over infinite traces.

\begin{lemma}
\label{lem:tlalcoured}
	%
	Partial and total $\LTLfALCOu$ formula satisfiability without the RDA are polynomial-time reducible to, respectively, partial and total $\LTLALCOu$ formula satisfiability without the RDA.
\end{lemma}
\begin{proof}
It is enough to observe that a $\LTLALCOu$ formula $\p$ can be transformed, 
by using the standard translation of temporal DLs into temporal first-order logic~\cite{GabEtAl03}, into an equisatisfiable $\QTLfr{\Until}{}{}$ formula
$\p'$.
Such a $\QTLfr{\Until}{}{}$ formula $\p'$ can be then mapped, by using the translation $\cdot^{\dagger}$
given in~\cite{ArtEtAl22},
into a $\QTLfr{\Until}{}{}$ formula $\p'^{\dagger}$ such that $\p'$ is satisfiable on finite, respectively partial or total, traces iff $\p'^{\dagger}$ is satisfiable on infinite traces, respectively partial or total (the notion of partial trace, given for  $\LTLALCOu$, can be naturally extended to the $\QTLfr{\Until}{}{}$ case).
\end{proof}

Given that $\LTLALCOu$ is a syntactic fragment of $\LTLALCOud$, as an
immediate consequence of Theorem~\ref{prop:tlalcoufintotundec} and the
reductions given in Lemmas~\ref{lemma:redtotaltopartial}
and~\ref{lem:tlalcoured}, we obtain the following result, holding for
total, and hence partial, interpretations.

\begin{theorem}
  $\LTLALCOu$,\nb{A: added} $\LTLALCOud$ and $\LTLfALCOud$ formula satisfiability
  without the RDA are undecidable.
\end{theorem}

Preserving\nb{A: added sentence} RDA is not enough to regain
decidability for the temporal DLs with definite descriptions. Indeed,
by exploiting definite descriptions that behave non-rigidly we can
still encode incrementation and decrementation of counters' values (in
place of non-rigid nominals), yielding the following result,
proved\nb{A: added} for total interpretations, and hence holding for
partial interpretations, too.

\begin{corollary}
$\LTLALCOud$ and $\LTLfALCOud$ formula satisfiability with the RDA is undecidable.
\end{corollary}


%% file: 6-conclusion.tex
\section{Discussion and Future Work}
\label{sec:conc}

\ifappendix
We moved first steps in the study of modal free description logics, focussing in particular on the epistemic free DL $\SfiveALCOud$, and on the
temporal free DLs $\LTLALCOud$ and $\LTLfALCOud$.
\fi
Syntactically, these DLs extend the classical $\ALCOu$, with nominals and the universal role, by including definite descriptions and
\ifappendix
modal
\fi
or temporal operators.
\ifappendix
Semantically, adapting ideas from the non-modal $\ALCOud$ case, we interpret these DLs over modal (epistemic or temporal) interpretations
that allow for:
(i) \emph{non-denoting} terms, that is, individual expressions that can be left uninterpreted at certain states, hence encoding ``error'' or ``empty'' values;
and
(ii) \emph{non-rigid} designators, making individual names and definite descriptions referential devices capable of picking different objects at different states, to represent ``dynamic'' or ``flexible'' values.
\fi
\ifappendix
The main technical results concern formula satisfiability: we show that, while this problem is $\NExpTime$-complete for $\SfiveALCOud$ on epistemic frames, it becomes undecidable 
both for
$\LTLALCOud$ and $\LTLfALCOud$
on, respectively, infinite and finite temporal traces.
\fi
%

\ifappendix
As future work, we intend to strengthen our results, as well as deepen the connections between free DLs with definite descriptions, on the one hand, and modal operators, on the other.
On the epistemic side, we are interested in:
\fi
\ifappendix
(i) considering frames for the propositional modal logics $\mathbf{K4}$, $\mathbf{T}$, $\mathbf{S4}$, or $\mathbf{KD45}$,
to model different doxastic or epistemic attitudes that might satisfy or not the so-called \emph{factivity} and \emph{introspection} principles~\cite{GabEtAl03};
\fi
\ifappendix
(ii) investigating of non-rigid descriptions and names
in the context of \emph{non-normal} modal DLs, such as the ones obtained from the systems $\mathbf{E}$, $\mathbf{M}$, $\mathbf{C}$, and $\mathbf{N}$~\cite{DalEtAl19,DalEtAl22,DalEtAl23},
to avoid the \emph{logical omniscience} problem (i.e., an agent knows all the logical truths and all the consequences of their background knowledge), which affects all the systems extending $\mathbf{K}$~\cite{Var86,Var89};
\fi
\ifappendix
(iii) addressing less expressive DL languages, such as $\ELOud$, in an epistemic setting, and connect them with the recently investigated \emph{standpoint DL} family~\cite{GomEtAl22,GomEtAl23}.
\fi

\ifappendix
On the temporal side, and in particular concerning our undecidability proofs, we observe the following.
\citeauthor{DegEtAl02}~\cite{DegEtAl02} use equality
to encode a one-unit addition or subtraction from a register's value at subsequent steps of a two-counter Minsky machine computation.
In our case, the same effect is due to the interplay between non-rigid (and possibly non-denoting) nominals and the universal role.
%
Related results appear also in~\citeauthor{HamKur15}~\cite{HamKur15},
where the undecidability of \emph{one-variable first-order temporal logic with counting to two}, denoted by FOLTL$^{\neq}$, is proved.
In their setting,  the \emph{difference} existential quantifier $\exists^{\neq}x$ extends the language of one-variable (monadic) first-order temporal logic: formulas like $\exists^{\neq}x \p$ enforce that $\p$ holds for some individual different from the one assigned to $x$.
Their undecidability proof for satisfiability of FOLTL$^{\neq}$ on finite or infinite traces with constant domains is also based on an encoding of the computation of two-counter Minsky machines.

It seems thus that a form of interaction between the universal role and counting constructors (even simply up to one, as in $\LTLALCOu$, with non-rigid and possibly non-denoting nominals) plays a role in these kinds of undecidability proofs.
However, it is not yet clear whether already $\LTLALCO$ or $\LTLALCOd$ satisfiability, on partial temporal interpretations and with non-rigid nominals, but \emph{without} the universal quantifier, is undecidable.
To obtain a similar proof, one would have to suitably replace the universal role with dedicated roles for counters, while forcing them to stay ``global'' enough not to lose any information about the values of the counters at subsequent computation steps.
We leave this open question for future work.

Finally, we remark that the preliminary negative results presented here for temporal free DLs might not entirely undermine the use of definite descriptions on a temporal dimension.
For applications in temporal conceptual modelling and ontology-mediated query answering~\cite{LutEtAl08,ArtEtAl17}, it is worth exploring whether more encouraging
results can be obtained in fragments restricting the use of temporal operators (limited, e.g., to the $\Box$ operator only), or constraining the DL dimension (such as the already mentioned $\LTLALCOd$, without the universal role, or lightweight members of the $\TDLite$ family~\cite{ArtEtAl14}).
Any of these restrictions should, of course, be designed to preserve an interesting degree of interaction between the DL constructors and the non-rigidity of names and descriptions.
\fi

Finally, we are interested in studying in this setting the complexities of other problems than formula satisfiability. Related to \emph{interpolant} and \emph{explicit definition existence}~\cite{ArtEtAl21,ArtEtAl23,KurEtAl23},
\ifappendix
the \emph{referring expression existence} problem~\cite{ArtEtAl21a},
which asks whether there exists a description for an individual under an ontology and a given signature,
is of particular interest to epistemic and temporal scenarios with definite descriptions.
\fi

\section*{Acknowledgements}
Andrea Mazzullo acknowledges the support of the MUR PNRR project FAIR - Future AI Research (PE00000013) funded by the NextGenerationEU.


%% file: main_arxiv_dl23.bbl
\begin{thebibliography}{37}
\expandafter\ifx\csname natexlab\endcsname\relax\def\natexlab#1{#1}\fi
\providecommand{\url}[1]{\texttt{#1}}
\providecommand{\href}[2]{#2}
\providecommand{\path}[1]{#1}
\providecommand{\DOIprefix}{doi:}
\providecommand{\ArXivprefix}{arXiv:}
\providecommand{\URLprefix}{URL: }
\providecommand{\Pubmedprefix}{pmid:}
\providecommand{\doi}[1]{\href{http://dx.doi.org/#1}{\path{#1}}}
\providecommand{\Pubmed}[1]{\href{pmid:#1}{\path{#1}}}
\providecommand{\bibinfo}[2]{#2}
\ifx\xfnm\relax \def\xfnm[#1]{\unskip,\space#1}\fi
\bibitem[{Lamport(1986)}]{Lamport:LaTeX}
\bibinfo{author}{L.~Lamport}, \bibinfo{title}{{\LaTeX}: A Document Preparation
  System}, \bibinfo{publisher}{Addison-Wesley}, \bibinfo{address}{Reading,
  MA.}, \bibinfo{year}{1986}.
\bibitem[{Zolin(2013)}]{zolin2013complexity}
\bibinfo{author}{E.~Zolin}, \bibinfo{title}{Complexity of reasoning in
  description logics}, \bibinfo{howpublished}{Available online},
  \bibinfo{year}{2013}. \URLprefix \url{http://www. cs. man. ac. uk/~
  ezolin/logic/complexity. html}.
\bibitem[{Schild(1991)}]{Schild1991}
\bibinfo{author}{K.~Schild},
\newblock \bibinfo{title}{A correspondence theory for terminological logics:
  {P}reliminary report},
\newblock in: \bibinfo{editor}{J.~Mylopoulos}, \bibinfo{editor}{R.~Reiter}
  (Eds.), \bibinfo{booktitle}{Proc.\ of the 12th Int.\ Joint Conf.\ on
  Artificial Intelligence (IJCAI~1991)}, \bibinfo{publisher}{Morgan Kaufmann},
  \bibinfo{year}{1991}, pp. \bibinfo{pages}{466--471}.
\bibitem[{Donini and Massacci(2000)}]{DoniniMassacci2000}
\bibinfo{author}{F.~M. Donini}, \bibinfo{author}{F.~Massacci},
\newblock \bibinfo{title}{Exptime tableaux for $\mathcal{ALC}$},
\newblock \bibinfo{journal}{Artificial Intelligence} \bibinfo{volume}{124}
  (\bibinfo{year}{2000}) \bibinfo{pages}{87--138}.
  \DOIprefix\doi{10.1016/S0004-3702(00)00070-9}.
\bibitem[{Knuth(1997)}]{Knuth97}
\bibinfo{author}{D.~E. Knuth}, \bibinfo{title}{The Art of Computer Programming,
  Vol. 1: Fundamental Algorithms (3rd. ed.)}, \bibinfo{publisher}{Addison
  Wesley Longman Publishing Co., Inc.}, \bibinfo{year}{1997}.
\bibitem[{Gundy et~al.(2007)Gundy, Balzarotti, and Vigna}]{VanGundy07}
\bibinfo{author}{M.~V. Gundy}, \bibinfo{author}{D.~Balzarotti},
  \bibinfo{author}{G.~Vigna},
\newblock \bibinfo{title}{Catch me, if you can: Evading network signatures with
  web-based polymorphic worms},
\newblock in: \bibinfo{booktitle}{Proceedings of the first USENIX workshop on
  Offensive Technologies}, WOOT '07, \bibinfo{publisher}{USENIX Association},
  \bibinfo{address}{Berkley, CA}, \bibinfo{year}{2007}.
\bibitem[{Abril and Plant(2007)}]{Abril07}
\bibinfo{author}{P.~S. Abril}, \bibinfo{author}{R.~Plant},
\newblock \bibinfo{title}{The patent holder's dilemma: Buy, sell, or troll?},
\newblock \bibinfo{journal}{Communications of the ACM} \bibinfo{volume}{50}
  (\bibinfo{year}{2007}) \bibinfo{pages}{36--44}.
  \DOIprefix\doi{10.1145/1188913.1188915}.
\bibitem[{Cohen et~al.(2007)Cohen, Nutt, and Sagic}]{Cohen07}
\bibinfo{author}{S.~Cohen}, \bibinfo{author}{W.~Nutt},
  \bibinfo{author}{Y.~Sagic},
\newblock \bibinfo{title}{Deciding equivalances among conjunctive aggregate
  queries},
\newblock \bibinfo{journal}{J. ACM} \bibinfo{volume}{54}
  (\bibinfo{year}{2007}). \DOIprefix\doi{10.1145/1219092.1219093}.
\bibitem[{Cohen(1996)}]{JCohen96}
\bibinfo{editor}{J.~Cohen} (Ed.), \bibinfo{title}{Special issue: Digital
  Libraries}, volume~\bibinfo{volume}{39}, \bibinfo{year}{1996}.
\bibitem[{Kosiur(2001)}]{Kosiur01}
\bibinfo{author}{D.~Kosiur}, \bibinfo{title}{Understanding Policy-Based
  Networking}, \bibinfo{edition}{2nd.} ed., \bibinfo{publisher}{Wiley},
  \bibinfo{address}{New York, NY}, \bibinfo{year}{2001}.
\bibitem[{Harel(1979)}]{Harel79}
\bibinfo{author}{D.~Harel}, \bibinfo{title}{First-Order Dynamic Logic},
  volume~\bibinfo{volume}{68} of \textit{\bibinfo{series}{Lecture Notes in
  Computer Science}}, \bibinfo{publisher}{Springer-Verlag},
  \bibinfo{address}{New York, NY}, \bibinfo{year}{1979}.
  \DOIprefix\doi{10.1007/3-540-09237-4}.
\bibitem[{Editor(2007)}]{Editor00}
\bibinfo{editor}{I.~Editor} (Ed.), \bibinfo{title}{The title of book one},
  volume~\bibinfo{volume}{9} of \textit{\bibinfo{series}{The name of the series
  one}}, \bibinfo{edition}{1st.} ed., \bibinfo{publisher}{University of Chicago
  Press}, \bibinfo{address}{Chicago}, \bibinfo{year}{2007}.
  \DOIprefix\doi{10.1007/3-540-09237-4}.
\bibitem[{Editor(2008)}]{Editor00a}
\bibinfo{editor}{I.~Editor} (Ed.), \bibinfo{title}{The title of book two}, The
  name of the series two, \bibinfo{edition}{2nd.} ed.,
  \bibinfo{publisher}{University of Chicago Press}, \bibinfo{address}{Chicago},
  \bibinfo{year}{2008}. \DOIprefix\doi{10.1007/3-540-09237-4}.
\bibitem[{Spector(1990)}]{Spector90}
\bibinfo{author}{A.~Z. Spector},
\newblock \bibinfo{title}{Achieving application requirements},
\newblock in: \bibinfo{editor}{S.~Mullender} (Ed.),
  \bibinfo{booktitle}{Distributed Systems}, \bibinfo{edition}{2nd.} ed.,
  \bibinfo{publisher}{ACM Press}, \bibinfo{address}{New York, NY},
  \bibinfo{year}{1990}, pp. \bibinfo{pages}{19--33}.
  \DOIprefix\doi{10.1145/90417.90738}.
\bibitem[{Douglass et~al.(1998)Douglass, Harel, and Trakhtenbrot}]{Douglass98}
\bibinfo{author}{B.~P. Douglass}, \bibinfo{author}{D.~Harel},
  \bibinfo{author}{M.~B. Trakhtenbrot},
\newblock \bibinfo{title}{Statecarts in use: structured analysis and
  object-orientation},
\newblock in: \bibinfo{editor}{G.~Rozenberg}, \bibinfo{editor}{F.~W.
  Vaandrager} (Eds.), \bibinfo{booktitle}{Lectures on Embedded Systems}, volume
  \bibinfo{volume}{1494} of \textit{\bibinfo{series}{Lecture Notes in Computer
  Science}}, \bibinfo{publisher}{Springer-Verlag}, \bibinfo{address}{London},
  \bibinfo{year}{1998}, pp. \bibinfo{pages}{368--394}.
  \DOIprefix\doi{10.1007/3-540-65193-4_29}.
\bibitem[{Andler(1979)}]{Andler79}
\bibinfo{author}{S.~Andler},
\newblock \bibinfo{title}{Predicate path expressions},
\newblock in: \bibinfo{booktitle}{Proceedings of the 6th. ACM SIGACT-SIGPLAN
  symposium on Principles of Programming Languages}, POPL '79,
  \bibinfo{publisher}{ACM Press}, \bibinfo{address}{New York, NY},
  \bibinfo{year}{1979}, pp. \bibinfo{pages}{226--236}.
  \DOIprefix\doi{10.1145/567752.567774}.
\bibitem[{Smith(2010)}]{Smith10}
\bibinfo{author}{S.~W. Smith},
\newblock \bibinfo{title}{An experiment in bibliographic mark-up: Parsing
  metadata for xml export},
\newblock in: \bibinfo{editor}{R.~N. Smythe}, \bibinfo{editor}{A.~Noble}
  (Eds.), \bibinfo{booktitle}{Proceedings of the 3rd. annual workshop on
  Librarians and Computers}, volume~\bibinfo{volume}{3} of
  \textit{\bibinfo{series}{LAC '10}}, \bibinfo{publisher}{Paparazzi Press},
  \bibinfo{address}{Milan Italy}, \bibinfo{year}{2010}, pp.
  \bibinfo{pages}{422--431}. \DOIprefix\doi{99.9999/woot07-S422}.
\bibitem[{Harel(1978)}]{Harel78}
\bibinfo{author}{D.~Harel}, \bibinfo{title}{LOGICS of Programs: AXIOMATICS and
  DESCRIPTIVE POWER}, \bibinfo{type}{MIT Research Lab Technical Report}
  \bibinfo{number}{TR-200}, Massachusetts Institute of Technology,
  \bibinfo{address}{Cambridge, MA}, \bibinfo{year}{1978}.
\bibitem[{Clarkson(1985)}]{Clarkson85}
\bibinfo{author}{K.~L. Clarkson}, \bibinfo{title}{Algorithms for Closest-Point
  Problems (Computational Geometry)}, Ph.D. thesis, Stanford University,
  \bibinfo{address}{Palo Alto, CA}, \bibinfo{year}{1985}. \bibinfo{note}{UMI
  Order Number: AAT 8506171}.
\bibitem[{Anisi(2003)}]{anisi03}
\bibinfo{author}{D.~A. Anisi}, \bibinfo{title}{Optimal Motion Control of a
  Ground Vehicle}, Master's thesis, Royal Institute of Technology (KTH),
  Stockholm, Sweden, \bibinfo{year}{2003}.
\bibitem[{Thornburg(2001)}]{Thornburg01}
\bibinfo{author}{H.~Thornburg}, \bibinfo{title}{Introduction to bayesian
  statistics}, \bibinfo{year}{2001}. \URLprefix
  \url{http://ccrma.stanford.edu/~jos/bayes/bayes.html}.
\bibitem[{Ablamowicz and Fauser(2007)}]{Ablamowicz07}
\bibinfo{author}{R.~Ablamowicz}, \bibinfo{author}{B.~Fauser},
  \bibinfo{title}{Clifford: a maple 11 package for clifford algebra
  computations, version 11}, \bibinfo{year}{2007}. \URLprefix
  \url{http://math.tntech.edu/rafal/cliff11/index.html}.
\bibitem[{Poker-Edge.Com(2006)}]{Poker06}
\bibinfo{author}{Poker-Edge.Com}, \bibinfo{title}{Stats and analysis},
  \bibinfo{year}{2006}. \URLprefix \url{http://www.poker-edge.com/stats.php}.
\bibitem[{Obama(2008)}]{Obama08}
\bibinfo{author}{B.~Obama}, \bibinfo{title}{A more perfect union},
  \bibinfo{howpublished}{Video}, \bibinfo{year}{2008}. \URLprefix
  \url{http://video.google.com/videoplay?docid=6528042696351994555}.
\bibitem[{Novak(2003)}]{Novak03}
\bibinfo{author}{D.~Novak},
\newblock \bibinfo{title}{Solder man},
\newblock in: \bibinfo{booktitle}{ACM SIGGRAPH 2003 Video Review on Animation
  theater Program: Part I - Vol. 145 (July 27--27, 2003)},
  \bibinfo{publisher}{ACM Press}, \bibinfo{address}{New York, NY},
  \bibinfo{year}{2003}, p.~\bibinfo{pages}{4}. \URLprefix
  \url{http://video.google.com/videoplay?docid=6528042696351994555}.
  \DOIprefix\doi{99.9999/woot07-S422}.
\bibitem[{Lee(2005)}]{Lee05}
\bibinfo{author}{N.~Lee},
\newblock \bibinfo{title}{Interview with bill kinder: January 13, 2005},
\newblock \bibinfo{journal}{Comput. Entertain.} \bibinfo{volume}{3}
  (\bibinfo{year}{2005}). \DOIprefix\doi{10.1145/1057270.1057278}.
\bibitem[{Scientist(2009)}]{JoeScientist001}
\bibinfo{author}{J.~Scientist}, \bibinfo{title}{The fountain of youth},
  \bibinfo{year}{2009}. \bibinfo{note}{Patent No. 12345, Filed July 1st., 2008,
  Issued Aug. 9th., 2009}.
\bibitem[{Rous(2008)}]{rous08}
\bibinfo{author}{B.~Rous},
\newblock \bibinfo{title}{The enabling of digital libraries},
\newblock \bibinfo{journal}{Digital Libraries} \bibinfo{volume}{12}
  (\bibinfo{year}{2008}). \bibinfo{note}{To appear}.
\bibitem[{Saeedi et~al.(2010{\natexlab{a}})Saeedi, Zamani, and
  Sedighi}]{SaeediMEJ10}
\bibinfo{author}{M.~Saeedi}, \bibinfo{author}{M.~S. Zamani},
  \bibinfo{author}{M.~Sedighi},
\newblock \bibinfo{title}{A library-based synthesis methodology for reversible
  logic},
\newblock \bibinfo{journal}{Microelectron. J.} \bibinfo{volume}{41}
  (\bibinfo{year}{2010}{\natexlab{a}}) \bibinfo{pages}{185--194}.
\bibitem[{Saeedi et~al.(2010{\natexlab{b}})Saeedi, Zamani, Sedighi, and
  Sasanian}]{SaeediJETC10}
\bibinfo{author}{M.~Saeedi}, \bibinfo{author}{M.~S. Zamani},
  \bibinfo{author}{M.~Sedighi}, \bibinfo{author}{Z.~Sasanian},
\newblock \bibinfo{title}{Synthesis of reversible circuit using cycle-based
  approach},
\newblock \bibinfo{journal}{J. Emerg. Technol. Comput. Syst.}
  \bibinfo{volume}{6} (\bibinfo{year}{2010}{\natexlab{b}}).
\bibitem[{Kirschmer and Voight(2010)}]{Kirschmer:2010:AEI:1958016.1958018}
\bibinfo{author}{M.~Kirschmer}, \bibinfo{author}{J.~Voight},
\newblock \bibinfo{title}{Algorithmic enumeration of ideal classes for
  quaternion orders},
\newblock \bibinfo{journal}{SIAM J. Comput.} \bibinfo{volume}{39}
  (\bibinfo{year}{2010}) \bibinfo{pages}{1714--1747}. \URLprefix
  \url{http://dx.doi.org/10.1137/080734467}. \DOIprefix\doi{10.1137/080734467}.
\bibitem[{H{\"o}rmander(1985{\natexlab{a}})}]{MR781536}
\bibinfo{author}{L.~H{\"o}rmander}, \bibinfo{title}{The analysis of linear
  partial differential operators. {IV}}, volume \bibinfo{volume}{275} of
  \textit{\bibinfo{series}{Grundlehren der Mathematischen Wissenschaften
  [Fundamental Principles of Mathematical Sciences]}},
  \bibinfo{publisher}{Springer-Verlag}, \bibinfo{address}{Berlin, Germany},
  \bibinfo{year}{1985}{\natexlab{a}}. \bibinfo{note}{Fourier integral
  operators}.
\bibitem[{H{\"o}rmander(1985{\natexlab{b}})}]{MR781537}
\bibinfo{author}{L.~H{\"o}rmander}, \bibinfo{title}{The analysis of linear
  partial differential operators. {III}}, volume \bibinfo{volume}{275} of
  \textit{\bibinfo{series}{Grundlehren der Mathematischen Wissenschaften
  [Fundamental Principles of Mathematical Sciences]}},
  \bibinfo{publisher}{Springer-Verlag}, \bibinfo{address}{Berlin, Germany},
  \bibinfo{year}{1985}{\natexlab{b}}. \bibinfo{note}{Pseudodifferential
  operators}.
\bibitem[{IEEE(2004)}]{2004:ITE:1009386.1010128}
IEEE,
\newblock \bibinfo{title}{Ieee tcsc executive committee},
\newblock in: \bibinfo{booktitle}{Proceedings of the IEEE International
  Conference on Web Services}, ICWS '04, \bibinfo{publisher}{IEEE Computer
  Society}, \bibinfo{address}{Washington, DC, USA}, \bibinfo{year}{2004}, pp.
  \bibinfo{pages}{21--22}. \DOIprefix\doi{10.1109/ICWS.2004.64}.
\bibitem[{TUG(2017)}]{TUGInstmem}
TUG, \bibinfo{title}{Institutional members of the {\TeX} users group},
  \bibinfo{year}{2017}. \URLprefix \url{http://www.tug.org/instmem.html}.
\bibitem[{{R Core Team}(2019)}]{R}
\bibinfo{author}{{R Core Team}}, \bibinfo{title}{R: A language and environment
  for statistical computing}, \bibinfo{year}{2019}. \URLprefix
  \url{https://www.R-project.org/}.
\bibitem[{Anzaroot and McCallum(2013)}]{UMassCitations}
\bibinfo{author}{S.~Anzaroot}, \bibinfo{author}{A.~McCallum},
  \bibinfo{title}{{UMass} citation field extraction dataset},
  \bibinfo{year}{2013}. \URLprefix
  \url{http://www.iesl.cs.umass.edu/data/data-umasscitationfield}.

\end{thebibliography}


\providecommand{\noopsort}[1]{}
\begin{thebibliography}{46}
\expandafter\ifx\csname natexlab\endcsname\relax\def\natexlab#1{#1}\fi
\providecommand{\url}[1]{\texttt{#1}}
\providecommand{\href}[2]{#2}
\providecommand{\path}[1]{#1}
\providecommand{\DOIprefix}{doi:}
\providecommand{\ArXivprefix}{arXiv:}
\providecommand{\URLprefix}{URL: }
\providecommand{\Pubmedprefix}{pmid:}
\providecommand{\doi}[1]{\href{http://dx.doi.org/#1}{\path{#1}}}
\providecommand{\Pubmed}[1]{\href{pmid:#1}{\path{#1}}}
\providecommand{\bibinfo}[2]{#2}
\ifx\xfnm\relax \def\xfnm[#1]{\unskip,\space#1}\fi
\bibitem[{Borgida et~al.(2016)Borgida, Toman, and Weddell}]{BorEtAl16}
\bibinfo{author}{A.~Borgida}, \bibinfo{author}{D.~Toman},
  \bibinfo{author}{G.~E. Weddell},
\newblock \bibinfo{title}{On referring expressions in query answering over
  first order knowledge bases},
\newblock in: \bibinfo{booktitle}{Proceedings of the 15th International
  Conference on Principles of Knowledge Representation and Reasoning
  ({KR}-16)}, \bibinfo{publisher}{{AAAI} Press}, \bibinfo{year}{2016}, pp.
  \bibinfo{pages}{319--328}.
\bibitem[{Borgida et~al.(2017)Borgida, Toman, and Weddell}]{BorEtAl17}
\bibinfo{author}{A.~Borgida}, \bibinfo{author}{D.~Toman},
  \bibinfo{author}{G.~E. Weddell},
\newblock \bibinfo{title}{Concerning referring expressions in query answers},
\newblock in: \bibinfo{booktitle}{Proceedings of the 26th International Joint
  Conference on Artificial Intelligence, ({IJCAI}-17)},
  \bibinfo{publisher}{ijcai.org}, \bibinfo{year}{2017}, pp.
  \bibinfo{pages}{4791--4795}.
\bibitem[{Toman and Weddell(2018)}]{TomWed18}
\bibinfo{author}{D.~Toman}, \bibinfo{author}{G.~E. Weddell},
\newblock \bibinfo{title}{Identity resolution in conjunctive querying over
  dl-based knowledge bases},
\newblock in: \bibinfo{booktitle}{Proceedings of the 31st International
  Workshop on Description Logics ({DL}-18)}, volume \bibinfo{volume}{2211} of
  \textit{\bibinfo{series}{{CEUR} Workshop Proceedings}},
  \bibinfo{publisher}{CEUR-WS.org}, \bibinfo{year}{2018}.
\bibitem[{Bencivenga(2002)}]{Ben02}
\bibinfo{author}{E.~Bencivenga},
\newblock \bibinfo{title}{Free logics},
\newblock in: \bibinfo{booktitle}{{Handbook of Philosophical Logic}},
  \bibinfo{publisher}{Springer}, \bibinfo{year}{2002}, pp.
  \bibinfo{pages}{147--196}.
\bibitem[{Lehmann(2002)}]{Leh02}
\bibinfo{author}{S.~Lehmann},
\newblock \bibinfo{title}{More free logic},
\newblock in: \bibinfo{booktitle}{{Handbook of Philosophical Logic}},
  \bibinfo{publisher}{Springer}, \bibinfo{year}{2002}, pp.
  \bibinfo{pages}{197--259}.
\bibitem[{Indrzejczak(2021)}]{Ind21}
\bibinfo{author}{A.~Indrzejczak},
\newblock \bibinfo{title}{Free logics are cut-free},
\newblock \bibinfo{journal}{Stud Logica} \bibinfo{volume}{109}
  (\bibinfo{year}{2021}) \bibinfo{pages}{859--886}.
\bibitem[{Indrzejczak and Zawidzki(2021)}]{IndZaw21}
\bibinfo{author}{A.~Indrzejczak}, \bibinfo{author}{M.~Zawidzki},
\newblock \bibinfo{title}{Tableaux for free logics with descriptions},
\newblock in: \bibinfo{editor}{A.~Das}, \bibinfo{editor}{S.~Negri} (Eds.),
  \bibinfo{booktitle}{Proceedings of the 30th International Conference on
  Automated Reasoning with Analytic Tableaux and Related Methods
  ({TABLEAUX}-21)}, volume \bibinfo{volume}{12842} of
  \textit{\bibinfo{series}{Lecture Notes in Computer Science}},
  \bibinfo{publisher}{Springer}, \bibinfo{year}{2021}, pp.
  \bibinfo{pages}{56--73}.
\bibitem[{Russell(1905)}]{Rus05}
\bibinfo{author}{B.~Russell},
\newblock \bibinfo{title}{{On Denoting}},
\newblock \bibinfo{journal}{Mind} \bibinfo{volume}{14} (\bibinfo{year}{1905})
  \bibinfo{pages}{479--493}.
\bibitem[{Neuhaus et~al.(2020)Neuhaus, Kutz, and Righetti}]{NeuEtAl20}
\bibinfo{author}{F.~Neuhaus}, \bibinfo{author}{O.~Kutz},
  \bibinfo{author}{G.~Righetti},
\newblock \bibinfo{title}{Free description logic for ontologists},
\newblock in: \bibinfo{booktitle}{Proceedings of the Joint Ontology Workshops
  ({JOWO}-20)}, volume \bibinfo{volume}{2708} of
  \textit{\bibinfo{series}{{CEUR} Workshop Proceedings}},
  \bibinfo{publisher}{CEUR-WS.org}, \bibinfo{year}{2020}.
\bibitem[{Artale et~al.(2020)Artale, Mazzullo, Ozaki, and Wolter}]{ArtEtAl20b}
\bibinfo{author}{A.~Artale}, \bibinfo{author}{A.~Mazzullo},
  \bibinfo{author}{A.~Ozaki}, \bibinfo{author}{F.~Wolter},
\newblock \bibinfo{title}{On free description logics with definite
  descriptions},
\newblock in: \bibinfo{booktitle}{{Proceedings of the 33rd International
  Workshop on Description Logics (DL-20)}}, volume \bibinfo{volume}{2663} of
  \textit{\bibinfo{series}{{CEUR} Workshop Proceedings}},
  \bibinfo{publisher}{CEUR-WS.org}, \bibinfo{year}{2020}.
\bibitem[{Artale et~al.(2021)Artale, Mazzullo, Ozaki, and Wolter}]{ArtEtAl21a}
\bibinfo{author}{A.~Artale}, \bibinfo{author}{A.~Mazzullo},
  \bibinfo{author}{A.~Ozaki}, \bibinfo{author}{F.~Wolter},
\newblock \bibinfo{title}{On free description logics with definite
  descriptions},
\newblock in: \bibinfo{booktitle}{{Proceedings of the 18th International
  Conference on Principles of Knowledge Representation and Reasoning (KR-21)}},
  \bibinfo{year}{2021}, pp. \bibinfo{pages}{63--73}.
\bibitem[{Fitting(2004)}]{Fit04}
\bibinfo{author}{M.~Fitting},
\newblock \bibinfo{title}{First-order intensional logic},
\newblock \bibinfo{journal}{Ann. Pure Appl. Log.} \bibinfo{volume}{127}
  (\bibinfo{year}{2004}) \bibinfo{pages}{171--193}.
\bibitem[{Fitting and Mendelsohn(2012)}]{FitMen12}
\bibinfo{author}{M.~Fitting}, \bibinfo{author}{R.~L. Mendelsohn},
  \bibinfo{title}{First-order {M}odal {L}ogic}, \bibinfo{publisher}{Springer
  Science \& Business Media}, \bibinfo{year}{2012}.
\bibitem[{Cocchiarella(1984)}]{Coc84}
\bibinfo{author}{N.~B. Cocchiarella},
\newblock \bibinfo{title}{Philosophical perspectives on quantification in tense
  and modal logic} \bibinfo{volume}{II: Extensions of Classical Logic}
  (\bibinfo{year}{1984}) \bibinfo{pages}{309--353}.
\bibitem[{Garson(2001)}]{Gar01}
\bibinfo{author}{J.~W. Garson},
\newblock \bibinfo{title}{Quantification in modal logic},
\newblock in: \bibinfo{booktitle}{Handbook of philosophical logic}, volume
  \bibinfo{volume}{II: Extensions of Classical Logic},
  \bibinfo{publisher}{Springer}, \bibinfo{year}{2001}, pp.
  \bibinfo{pages}{267--323}.
\bibitem[{Bra{\"u}ner and Ghilardi(2007)}]{BraGhi07}
\bibinfo{author}{T.~Bra{\"u}ner}, \bibinfo{author}{S.~Ghilardi},
\newblock \bibinfo{title}{{First-order Modal Logic}},
\newblock in: \bibinfo{booktitle}{Handbook of Modal Logic},
  \bibinfo{publisher}{Elsevier}, \bibinfo{year}{2007}, pp.
  \bibinfo{pages}{549--620}.
\bibitem[{Kr{\"{o}}ger and Merz(2008)}]{KroMerz08}
\bibinfo{author}{F.~Kr{\"{o}}ger}, \bibinfo{author}{S.~Merz},
  \bibinfo{title}{Temporal Logic and State Systems}, Texts in Theoretical
  Computer Science. An {EATCS} Series, \bibinfo{publisher}{Springer},
  \bibinfo{year}{2008}.
\bibitem[{Corsi and Orlandelli(2013)}]{CorOrl13}
\bibinfo{author}{G.~Corsi}, \bibinfo{author}{E.~Orlandelli},
\newblock \bibinfo{title}{Free quantified epistemic logics},
\newblock \bibinfo{journal}{Studia Logica} \bibinfo{volume}{101}
  (\bibinfo{year}{2013}) \bibinfo{pages}{1159--1183}.
\bibitem[{Indrzejczak(2020)}]{Ind20}
\bibinfo{author}{A.~Indrzejczak},
\newblock \bibinfo{title}{Existence, definedness and definite descriptions in
  hybrid modal logic},
\newblock in: \bibinfo{booktitle}{Proceedings of the 13th Conference on
  Advances in Modal Logic ({AiML}-20)}, \bibinfo{publisher}{College
  Publications}, \bibinfo{year}{2020}, pp. \bibinfo{pages}{349--368}.
\bibitem[{Orlandelli(2021)}]{Orl21}
\bibinfo{author}{E.~Orlandelli},
\newblock \bibinfo{title}{Labelled calculi for quantified modal logics with
  definite descriptions},
\newblock \bibinfo{journal}{J. Log. Comput.} \bibinfo{volume}{31}
  (\bibinfo{year}{2021}) \bibinfo{pages}{923--946}.
\bibitem[{Wolter and Zakharyaschev(1998)}]{WolZak98}
\bibinfo{author}{F.~Wolter}, \bibinfo{author}{M.~Zakharyaschev},
\newblock \bibinfo{title}{Temporalizing description logics},
\newblock in: \bibinfo{booktitle}{{Proceedings of the 2nd International
  Symposium on Frontiers of Combining Systems (FroCoS-98)}},
  \bibinfo{publisher}{Research Studies Press/Wiley}, \bibinfo{year}{1998}, pp.
  \bibinfo{pages}{104--109}.
\bibitem[{Artale and Franconi(2005)}]{ArtFra05}
\bibinfo{author}{A.~Artale}, \bibinfo{author}{E.~Franconi},
\newblock \bibinfo{title}{Temporal description logics},
\newblock in: \bibinfo{booktitle}{Handbook of Temporal Reasoning in Artificial
  Intelligence}, volume~\bibinfo{volume}{1} of
  \textit{\bibinfo{series}{Foundations of Artificial Intelligence}},
  \bibinfo{publisher}{Elsevier}, \bibinfo{year}{2005}, pp.
  \bibinfo{pages}{375--388}.
\bibitem[{Lutz et~al.(2008)Lutz, Wolter, and Zakharyaschev}]{LutEtAl08}
\bibinfo{author}{C.~Lutz}, \bibinfo{author}{F.~Wolter},
  \bibinfo{author}{M.~Zakharyaschev},
\newblock \bibinfo{title}{Temporal description logics: {A} survey},
\newblock in: \bibinfo{booktitle}{Proceedings of the 15th International
  Symposium on Temporal Representation and Reasoning ({TIME}-08)},
  \bibinfo{publisher}{{IEEE} Computer Society}, \bibinfo{year}{2008}, pp.
  \bibinfo{pages}{3--14}.
\bibitem[{Donini et~al.(1998)Donini, Lenzerini, Nardi, Nutt, and
  Schaerf}]{DonEtAl98}
\bibinfo{author}{F.~M. Donini}, \bibinfo{author}{M.~Lenzerini},
  \bibinfo{author}{D.~Nardi}, \bibinfo{author}{W.~Nutt},
  \bibinfo{author}{A.~Schaerf},
\newblock \bibinfo{title}{An epistemic operator for description logics},
\newblock \bibinfo{journal}{Artif. Intell.} \bibinfo{volume}{100}
  (\bibinfo{year}{1998}) \bibinfo{pages}{225--274}.
\bibitem[{Calvanese et~al.(2008)Calvanese, Giacomo, Lembo, Lenzerini, and
  Rosati}]{CalEtAl08}
\bibinfo{author}{D.~Calvanese}, \bibinfo{author}{G.~D. Giacomo},
  \bibinfo{author}{D.~Lembo}, \bibinfo{author}{M.~Lenzerini},
  \bibinfo{author}{R.~Rosati},
\newblock \bibinfo{title}{Inconsistency tolerance in {P2P} data integration: An
  epistemic logic approach},
\newblock \bibinfo{journal}{Inf. Syst.} \bibinfo{volume}{33}
  (\bibinfo{year}{2008}) \bibinfo{pages}{360--384}.
\bibitem[{Console and Lenzerini(2020)}]{ConLen20}
\bibinfo{author}{M.~Console}, \bibinfo{author}{M.~Lenzerini},
\newblock \bibinfo{title}{Epistemic integrity constraints for ontology-based
  data management},
\newblock in: \bibinfo{booktitle}{Proceedings of the 34th {AAAI} Conference on
  Artificial Intelligence ({AAAI}-20)}, \bibinfo{publisher}{{AAAI} Press},
  \bibinfo{year}{2020}, pp. \bibinfo{pages}{2790--2797}.
\bibitem[{Mehdi and Rudolph(2011)}]{MehRud11}
\bibinfo{author}{A.~Mehdi}, \bibinfo{author}{S.~Rudolph},
\newblock \bibinfo{title}{Revisiting semantics for epistemic extensions of
  description logics},
\newblock in: \bibinfo{booktitle}{Proceedings of the 25th {AAAI} Conference on
  Artificial Intelligence ({AAAI}-11)}, \bibinfo{publisher}{{AAAI} Press},
  \bibinfo{year}{2011}.
\bibitem[{Gabbay et~al.(2003)Gabbay, Kurucz, Wolter, and
  Zakharyaschev}]{GabEtAl03}
\bibinfo{author}{D.~M. Gabbay}, \bibinfo{author}{A.~Kurucz},
  \bibinfo{author}{F.~Wolter}, \bibinfo{author}{M.~Zakharyaschev},
  \bibinfo{title}{Many-dimensional Modal Logics: Theory and Applications},
  \bibinfo{publisher}{North Holland Publishing Company}, \bibinfo{year}{2003}.
\bibitem[{Baader et~al.(2003)Baader, Calvanese, McGuinness, Nardi, and
  Patel-Schneider}]{BaaEtAl03a}
\bibinfo{editor}{F.~Baader}, \bibinfo{editor}{D.~Calvanese},
  \bibinfo{editor}{D.~L. McGuinness}, \bibinfo{editor}{D.~Nardi},
  \bibinfo{editor}{P.~F. Patel-Schneider} (Eds.), \bibinfo{title}{The
  Description Logic Handbook: Theory, Implementation, and Applications},
  \bibinfo{publisher}{Cambridge University Press}, \bibinfo{year}{2003}.
\bibitem[{Rudolph(2011)}]{Rud11}
\bibinfo{author}{S.~Rudolph},
\newblock \bibinfo{title}{Foundations of description logics},
\newblock in: \bibinfo{booktitle}{Tutorial Lectures of the 7th International
  Summer School 2011 on Reasoning Web}, volume \bibinfo{volume}{6848} of
  \textit{\bibinfo{series}{Lecture Notes in Computer Science}},
  \bibinfo{publisher}{Springer}, \bibinfo{year}{2011}, pp.
  \bibinfo{pages}{76--136}.
\bibitem[{Degtyarev et~al.(2002)Degtyarev, Fisher, and Lisitsa}]{DegEtAl02}
\bibinfo{author}{A.~Degtyarev}, \bibinfo{author}{M.~Fisher},
  \bibinfo{author}{A.~Lisitsa},
\newblock \bibinfo{title}{Equality and monodic first-order temporal logic},
\newblock \bibinfo{journal}{Studia Logica} \bibinfo{volume}{72}
  (\bibinfo{year}{2002}) \bibinfo{pages}{147--156}.
\bibitem[{Baader et~al.(2017)Baader, Horrocks, Lutz, and Sattler}]{BaaEtAl17}
\bibinfo{author}{F.~Baader}, \bibinfo{author}{I.~Horrocks},
  \bibinfo{author}{C.~Lutz}, \bibinfo{author}{U.~Sattler}, \bibinfo{title}{An
  Introduction to Description Logic}, \bibinfo{publisher}{Cambridge University
  Press}, \bibinfo{year}{2017}.
\bibitem[{Artale et~al.(2022)Artale, Mazzullo, and Ozaki}]{ArtEtAl22}
\bibinfo{author}{A.~Artale}, \bibinfo{author}{A.~Mazzullo},
  \bibinfo{author}{A.~Ozaki},
\newblock \bibinfo{title}{First-order temporal logic on finite traces: Semantic
  properties, decidable fragments, and applications},
\newblock \bibinfo{journal}{CoRR} \bibinfo{volume}{abs/2202.00610}
  (\bibinfo{year}{2022}).
\bibitem[{Dalmonte et~al.(2019)Dalmonte, Mazzullo, and Ozaki}]{DalEtAl19}
\bibinfo{author}{T.~Dalmonte}, \bibinfo{author}{A.~Mazzullo},
  \bibinfo{author}{A.~Ozaki},
\newblock \bibinfo{title}{On non-normal modal description logics},
\newblock in: \bibinfo{editor}{M.~Simkus}, \bibinfo{editor}{G.~E. Weddell}
  (Eds.), \bibinfo{booktitle}{{DL}}, volume \bibinfo{volume}{2373},
  \bibinfo{publisher}{CEUR-WS.org}, \bibinfo{year}{2019}.
\bibitem[{Dalmonte et~al.(2022)Dalmonte, Mazzullo, and Ozaki}]{DalEtAl22}
\bibinfo{author}{T.~Dalmonte}, \bibinfo{author}{A.~Mazzullo},
  \bibinfo{author}{A.~Ozaki},
\newblock \bibinfo{title}{Reasoning in non-normal modal description logics},
\newblock in: \bibinfo{editor}{C.~Benzm{\"{u}}ller}, \bibinfo{editor}{J.~Otten}
  (Eds.), \bibinfo{booktitle}{{ARQNL}@{IJCAR}}, volume \bibinfo{volume}{2095},
  \bibinfo{year}{2022}, pp. \bibinfo{pages}{28--45}.
\bibitem[{Dalmonte et~al.(2023)Dalmonte, Mazzullo, Ozaki, and
  Troquard}]{DalEtAl23}
\bibinfo{author}{T.~Dalmonte}, \bibinfo{author}{A.~Mazzullo},
  \bibinfo{author}{A.~Ozaki}, \bibinfo{author}{N.~Troquard},
\newblock \bibinfo{title}{Non-normal modal description logics},
\newblock in: \bibinfo{booktitle}{Proceedings of the 18th European Conference
  on Logics in Artificial Intelligence ({JELIA}-23), to appear},
  \bibinfo{year}{2023}.
\bibitem[{Vardi(1986)}]{Var86}
\bibinfo{author}{M.~Y. Vardi},
\newblock \bibinfo{title}{On epistemic logic and logical omniscience},
\newblock in: \bibinfo{editor}{J.~Y. Halpern} (Ed.),
  \bibinfo{booktitle}{Proceedings of the 1st Conference on Theoretical Aspects
  of Reasoning about Knowledge ({TARK}-86)}, \bibinfo{publisher}{Morgan
  Kaufmann}, \bibinfo{year}{1986}, pp. \bibinfo{pages}{293--305}.
\bibitem[{Vardi(1989)}]{Var89}
\bibinfo{author}{M.~Y. Vardi},
\newblock \bibinfo{title}{On the complexity of epistemic reasoning},
\newblock in: \bibinfo{booktitle}{Proceedings of the 4th Annual Symposium on
  Logic in Computer Science ({LICS}-89)}, \bibinfo{publisher}{{IEEE} Computer
  Society}, \bibinfo{year}{1989}, pp. \bibinfo{pages}{243--252}.
\bibitem[{{Gómez Álvarez} et~al.(2022){Gómez Álvarez}, Rudolph, and
  Strass}]{GomEtAl22}
\bibinfo{author}{L.~{Gómez Álvarez}}, \bibinfo{author}{S.~Rudolph},
  \bibinfo{author}{H.~Strass},
\newblock \bibinfo{title}{How to agree to disagree - managing ontological
  perspectives using standpoint logic},
\newblock in: \bibinfo{booktitle}{Proceedings of the 21st International
  Semantic Web Conference ({ISWC}-22)}, volume \bibinfo{volume}{13489},
  \bibinfo{year}{2022}, pp. \bibinfo{pages}{125--141}.
\bibitem[{{Gómez Álvarez} et~al.(2023){Gómez Álvarez}, Rudolph, and
  Strass}]{GomEtAl23}
\bibinfo{author}{L.~{Gómez Álvarez}}, \bibinfo{author}{S.~Rudolph},
  \bibinfo{author}{H.~Strass},
\newblock \bibinfo{title}{Tractable diversity: Scalable multiperspective
  ontology management via standpoint {EL}},
\newblock \bibinfo{journal}{CoRR} \bibinfo{volume}{abs/2302.13187}
  (\bibinfo{year}{2023}).
\bibitem[{Hampson and Kurucz(2015)}]{HamKur15}
\bibinfo{author}{C.~Hampson}, \bibinfo{author}{A.~Kurucz},
\newblock \bibinfo{title}{Undecidable propositional bimodal logics and
  one-variable first-order linear temporal logics with counting},
\newblock \bibinfo{journal}{{ACM} Trans. Comput. Log.} \bibinfo{volume}{16}
  (\bibinfo{year}{2015}) \bibinfo{pages}{27:1--27:36}.
\bibitem[{Artale et~al.(2017)Artale, Kontchakov, Kovtunova, Ryzhikov, Wolter,
  and Zakharyaschev}]{ArtEtAl17}
\bibinfo{author}{A.~Artale}, \bibinfo{author}{R.~Kontchakov},
  \bibinfo{author}{A.~Kovtunova}, \bibinfo{author}{V.~Ryzhikov},
  \bibinfo{author}{F.~Wolter}, \bibinfo{author}{M.~Zakharyaschev},
\newblock \bibinfo{title}{Ontology-mediated query answering over temporal data:
  {A} survey (invited talk)},
\newblock in: \bibinfo{booktitle}{{Proceedings of the 24th International
  Symposium on Temporal Representation and Reasoning, (TIME-17)}},
  volume~\bibinfo{volume}{90} of \textit{\bibinfo{series}{LIPIcs}},
  \bibinfo{publisher}{Schloss Dagstuhl - Leibniz-Zentrum f{\"{u}}r Informatik},
  \bibinfo{year}{2017}, pp. \bibinfo{pages}{1:1--1:37}.
\bibitem[{Artale et~al.(2014)Artale, Kontchakov, Ryzhikov, and
  Zakharyaschev}]{ArtEtAl14}
\bibinfo{author}{A.~Artale}, \bibinfo{author}{R.~Kontchakov},
  \bibinfo{author}{V.~Ryzhikov}, \bibinfo{author}{M.~Zakharyaschev},
\newblock \bibinfo{title}{A cookbook for temporal conceptual data modelling
  with description logics},
\newblock \bibinfo{journal}{{ACM} Trans. Comput. Log.} \bibinfo{volume}{15}
  (\bibinfo{year}{2014}) \bibinfo{pages}{25:1--25:50}.
\bibitem[{Artale et~al.(2021)Artale, Jung, Mazzullo, Ozaki, and
  Wolter}]{ArtEtAl21}
\bibinfo{author}{A.~Artale}, \bibinfo{author}{J.~C. Jung},
  \bibinfo{author}{A.~Mazzullo}, \bibinfo{author}{A.~Ozaki},
  \bibinfo{author}{F.~Wolter},
\newblock \bibinfo{title}{Living without beth and craig: Definitions and
  interpolants in description logics with nominals and role inclusions},
\newblock in: \bibinfo{booktitle}{{Proceedings of the 35th AAAI Conference on
  Artificial Intelligence (AAAI-21)}}, \bibinfo{publisher}{{AAAI} Press},
  \bibinfo{year}{2021}, pp. \bibinfo{pages}{6193--6201}.
\bibitem[{Artale et~al.(2023)Artale, Jung, ~, Ozaki, and Wolter}]{ArtEtAl23}
\bibinfo{author}{A.~Artale}, \bibinfo{author}{J.~C. Jung},
  \bibinfo{author}{A.~, Mazzullo}, \bibinfo{author}{A.~Ozaki},
  \bibinfo{author}{F.~Wolter},
\newblock \bibinfo{title}{Living without beth and craig: Definitions and
  interpolants in description and modal logics with nominals and role
  inclusions},
\newblock \bibinfo{journal}{{ACM} Trans. Comput. Log.} \bibinfo{volume}{Online
  (Just Accepted)} (\bibinfo{year}{2023}).
\bibitem[{Kurucz et~al.(2023)Kurucz, Wolter, and Zakharyaschev}]{KurEtAl23}
\bibinfo{author}{A.~Kurucz}, \bibinfo{author}{F.~Wolter},
  \bibinfo{author}{M.~Zakharyaschev},
\newblock \bibinfo{title}{Definitions and (uniform) interpolants in first-order
  modal logic},
\newblock \bibinfo{journal}{CoRR} \bibinfo{volume}{abs/2303.04598}
  (\bibinfo{year}{2023}).

\end{thebibliography}
